\documentclass[preprint,review,12pt]{elsarticle}
\usepackage{lineno}
\usepackage{graphicx}
\usepackage{amsmath}
\usepackage{amssymb}
\usepackage{setspace}
\usepackage{enumerate}
\usepackage{microtype}
\usepackage{xcolor}
\usepackage{hyperref}
\definecolor{deepgreen}{rgb}{0,0.5,0}
\hypersetup{
    colorlinks=true,
    citecolor={deepgreen}
}
\usepackage{url}

\usepackage{bm}
\usepackage{graphicx}
\usepackage{enumerate}
\usepackage{microtype}
\usepackage{multirow}
\usepackage{array}

\usepackage[ruled,linesnumbered,vlined]{algorithm2e}

\newtheorem{theorem}{Theorem}

\newcolumntype{C}[1]{>{\centering\let\newline\\\arraybackslash\hspace{0pt}}m{#1}}
\usepackage{lineno}
\journal{Journal}

\begin{document}

\begin{frontmatter}

\title{Low-distortion planar embedding of rod-based structures}

\author[label1]{Mark Yan Lok Yip}

\author[label1]{Gary P. T. Choi\corref{cor1}}
\cortext[cor1]{Corresponding author.}

\address[label1]{Department of Mathematics, The Chinese University of Hong Kong}

\begin{abstract}
 Rod-based structures are commonly used in practical applications in science and engineering. However, in many design, analysis, and manufacturing tasks, handling the rod-based structures in three dimensions directly is generally challenging. To simplify the tasks, it is usually more desirable to achieve a two-dimensional representation of the rod-based structures via some suitable geometric mappings. In this work, we develop a novel method for computing a low-distortion planar embedding of rod-based structures. Specifically, we identify geometrical constraints that aim to preserve key length and angle quantities of the 3D rod-based structures and prevent the occurrence of overlapping rods in the planar embedding. Experimental results with a variety of rod-based structures are presented to demonstrate the effectiveness of our approach. Moreover, our method can be naturally extended to the design and mapping of hybrid structures consisting of both rods and surface elements. Altogether, our approach paves a new way for the efficient design and fabrication of novel three-dimensional geometric structures for practical applications.

\end{abstract}


\end{frontmatter}

\section{Introduction}
The design and analysis of geometric structures are important in many applications in science and engineering. In recent years, there has been an increasing interest in the use of wireframe or gridshell structures. For instance, structures consisting of strip elements can be utilized in the design of complex shapes in architecture~\cite{wang2023rectifying}. Structures consisting of rod-like elements are also commonly found in 3D or 4D printing~\cite{sydney2016biomimetic}. The high flexibility in such structures also makes them very suitable for achieving different shape transformation effects~\cite{risso2022highly,liu2023deployable}. Therefore, many recent approaches have focused on the design of different rod-based and gridshell structures, such as the elastic and thermoelastic materials~\cite{wang2004level}, elastic gridshells~\cite{baek2018form,qin2020genetic}, self-actuated shells for morphing into a prescribed 3D shape~\cite{guseinov2020programming}, surface discretization using strips~\cite{martin2021surface}, surface-based inflatables~\cite{panetta2021computational,ren2024computational}, and hybrid gridshell structures~\cite{schling2022designing}.

While these 3D rod-based and gridshell structures are highly flexible and widely applicable to different problems, the design and analysis of them require the consideration of various geometric aspects. Also, from the perspective of fabrication, storage, and transportation, directly handling the 3D structures may be inefficient. To simplify the process, one possible approach is to flatten the 3D structures and embed them on a planar domain. More specifically, the 2D representation of the structures will not only allow for easier design and comparison but also facilitate the practical manufacturing of the structures. However, to ensure that the 2D representation can accurately represent the original 3D rod-based structures, it is important to preserve certain key geometric properties of the 3D structures and minimize the overall geometric distortion of the embedding. Note that this core problem of computing a bijective and low-distortion planar mapping is closely related to domain and mesh parameterization~\cite{floater2005surface,sheffer2007mesh}. In particular, many parameterization methods have been developed based on conformal mapping~\cite{levy2002least,desbrun2002intrinsic,mullen2008spectral,jin2008discrete,choi2021efficient}, authalic mapping~\cite{desbrun2002intrinsic,zou2011authalic,zhao2013area,choi2018density}, as-rigid-as-possible mapping~\cite{sorkine2007rigid,liu2008local,wang2018novel}, quasi-conformal mapping~\cite{choi2022recent,pan2022constructing,pan2023g1}, and multi-patch methods~\cite{kapl2018construction,zou2025mat}. However, the nature of the 3D rod-based structures considered in this work presents a unique challenge. In particular, many existing parameterization methods consider the discrete meshes or structures as a representation of an underlying smooth surface in their methodological development and algorithmic design, while the 3D rod-based structures are solely a collection of vertices and rod segments without such an assumption. Hence, most of the relevant optimization formulations and computational schemes commonly used in prior parameterization methods are not directly applicable to the problem here.

\begin{figure}[t]
    \centering
    \includegraphics[width=\linewidth]{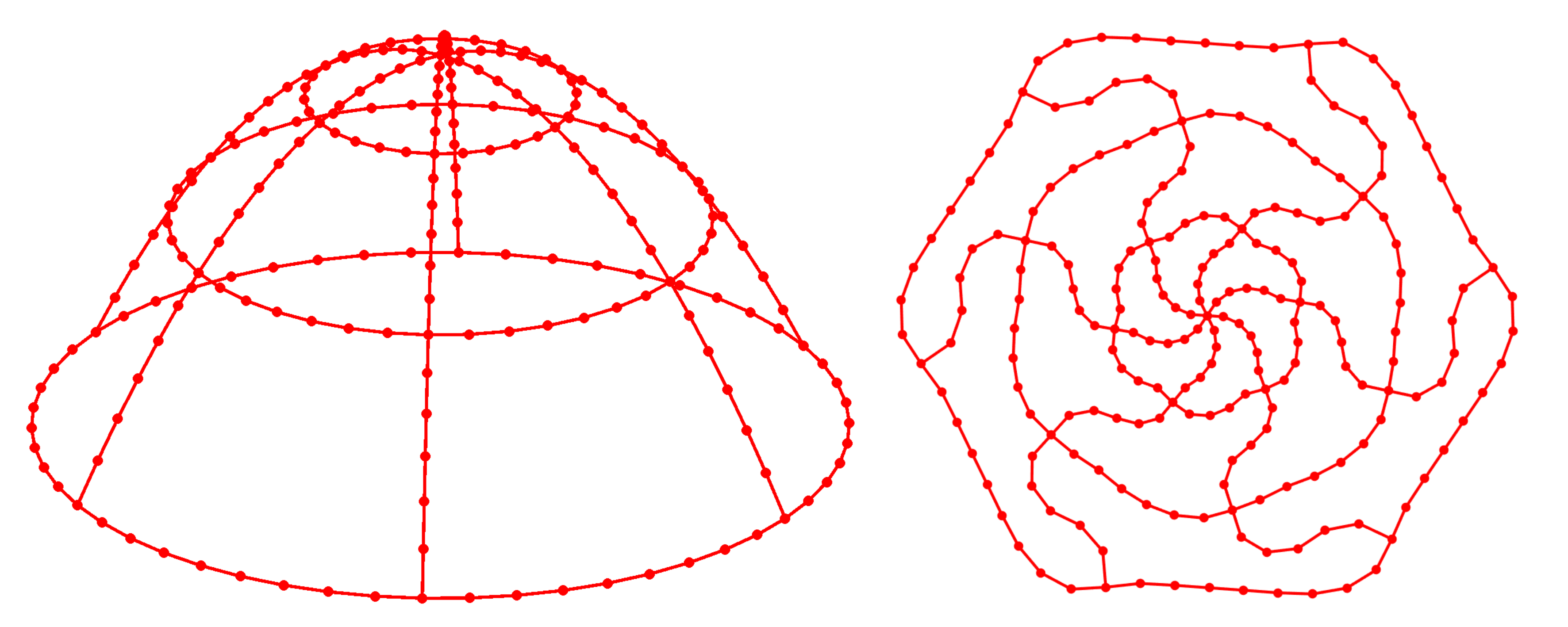}
    \caption{\textbf{An illustration of the proposed low-distortion planar embedding method for 3D rod-based structures in our work.} Given a 3D rod-based structure with a prescribed target shape (left), our proposed method aims to produce a planar configuration of the structure that possesses low geometric distortion (right) with key rod length and angle quantities preserved, thereby facilitating the modelling and manufacturing of the rod-based structures for practical applications. }
    \label{fig:illustration}
\end{figure}

In this work, we develop a novel method for achieving a low-distortion planar embedding of rod-based structures (Fig.~\ref{fig:illustration}). More specifically, our proposed method takes an arbitrary 3D rod-based structure consisting of vertices and short rod segments as input and produces a planar representation of it with minimal geometric distortion, preserving the lengths of the short rod segments and key intersection angles between rods. The method also automatically detects and corrects overlaps between rod segments to ensure a one-to-one correspondence between the planar representation and the original 3D structure. Experimental results on a large variety of rod-based structures are presented to demonstrate the effectiveness of our proposed method. We further extend the method for handling more complex hybrid structures and consider the 2D-to-3D morphing process to showcase the flexibility and applicability of our method.

The rest of the paper is organized as follows. In Section~\ref{sect:formulation}, we describe the formulation and algorithm of our proposed low-distortion planar embedding method, including the details of the optimization problem we consider and the corresponding objective and constraint functions. In Section~\ref{sect:experiment}, we test our proposed method on a large variety of rod-based structures with different geometric properties to demonstrate its effectiveness. In Section~\ref{sect:extension}, we further extend our computational framework for hybrid structures consisting of not only one-dimensional rods but also some surface regions. In Section~\ref{sect:morphing}, we simulate the 2D-to-3D morphing process of our planar embedding results to verify their practicability in real-world applications. We conclude the paper and discuss future directions in Section~\ref{sect:conclusion}.

\section{Proposed method} \label{sect:formulation}

Let $\mathcal{S} = (\mathcal{V}, \mathcal{E})$ be a 3D rod-based structure, where $\mathcal{V}$ is the vertex set with $m$ vertices in it and $\mathcal{E}$ is the edge set with $p$ edges in it. Each vertex $v_i = (x_i, y_i, z_i) \in \mathcal{V}$ is a node in $\mathbb{R}^3$, and each (undirected) edge $e_{i} = [e_i(1), e_i(2)] \in \mathcal{E}$ represents a short rod segment connecting node $v_{e_i(1)}$ and node $v_{e_i(2)}$. Our goal is to obtain a low-distortion planar embedding $f: \mathcal{S} \to \mathbb{R}^m \times \mathbb{R}^m$ that preserves the key geometric features of $\mathcal{S}$ including the lengths of the short rod segments and the intersection angles between rods as much as possible. Mechanically, the edges in the rod-based structures we consider are rigid, straight rod segments. We consider the joints between the rod segments as spherical (ball-and-socket) joints, which allow two connected rod segments to rotate freely relative to each other. Under the planar embedding, the rod-based structure should be flattened without any stretching or bending of the rod segments and without any disconnections of the joints. We further define the major joints as the joints that are the intersection points of more than two rod segments. Because of the high mechanical complexity of such major joints, we will further impose geometric constraints and preserve the angles at them in our embedding. The set $\mathcal{P} = f(\mathcal{S})$ will exactly be the vertex set of the planar embedding, and hence $(\mathcal{P},\mathcal{E_{\text{2D}}})$ will be the desired 2D representation of the rod-based structure where $\mathcal{E_{\text{2D}}}$ contains all edges in the form of $[f(v_i), f(v_j)]$ with $[v_i, v_j] \in \mathcal{E}$.

\subsection{Initial embedding}
The first step of our proposed method is to construct an initial embedding $f_0: \mathcal{S} \to \mathbb{R}^{m} \times \mathbb{R}^{m}$ of the given 3D rod-based structure. With the initial embedding, the subsequent computations can then be handled as a 2D-to-2D optimization problem more easily, as we can fix the dimension of the structure in $\mathbb{R}^2$ during the optimization process. 

To achieve this initial embedding, one possible approach is to compute a Tutte embedding~\cite{tutte1963draw}. More specifically, we solve the following equation
\begin{equation}
    L \mathbf{v} = 0,
\end{equation}
where $\mathbf{v}$ are the desired 2D coordinates of the vertices in the initial embedding and $L$ is a $m\times m$ matrix known as the graph Laplacian:
\begin{equation}
    L_{ij} = \left\{\begin{array}{cl}
    1, & \text{if } \space[v_i,v_j]\in \mathcal{E},\\
    -\sum\limits_{\substack{k=1, \ k\neq i}}^{m}L_{ik}, & \text{if }\space i=j,\\
    0, & \text{otherwise,}
    \end{array} \right.
\end{equation}
subject to some prescribed constraints $f_0(v_{bdy_i}) = w_i$ for all outermost vertices $v_{bdy_1}, v_{bdy_2}, \dots, v_{bdy_k}$, with $w_i$ being some prescribed position for the vertex $v_{bdy_i}$ for this initial step. By solving the above-mentioned equation, we can obtain a set of vertex coordinates $\mathcal{P}_0 = f_0(\mathcal{S})$ on the plane, which serves as our initialization. 

We remark that besides the Tutte embedding method, one may also consider other methods for efficiently obtaining the initial embedding. For instance, for some rod-based structures with relatively simple geometry, one may directly compute a projection $(x,y,z) \to (x,y)$ for the initial embedding. It may also be possible to combine some mesh flattening methods with other planar transformations to construct the initial embedding.

\subsection{Shape optimization}
Once we have obtained the initial embedding $f_0: \mathcal{S} \to \mathcal{P}_0$, we focus on solving a planar shape optimization problem and search for an optimized planar mapping $g: \mathcal{P}_0 \to \mathcal{P}$ that satisfies certain geometrical constraints. Below, we describe the geometrical constraints and their formulations in detail. Also, to enhance the computational efficiency and accuracy of the optimization process, we further derive the explicit formulas of the Jacobian of all constraints and objective functions.

First, note that to achieve a meaningful and accurate planar representation of the 3D rod-based structures, it is natural to consider the distortion of the rod segments and their joints. Specifically, it is desired that the planar shape in the optimized form should preserve the rod segment lengths and the angles at the major joints of rods. Also, there should be no overlaps between the rods in the planar representation. This motivates us to consider three major constraints, namely (i) the \emph{length-preserving constraints}, (ii) the \emph{angle-preserving constraints}, and (iii) the \emph{no-overlap constraints}.

\begin{figure}[t]
    \centering
    \includegraphics[width=0.75\linewidth]{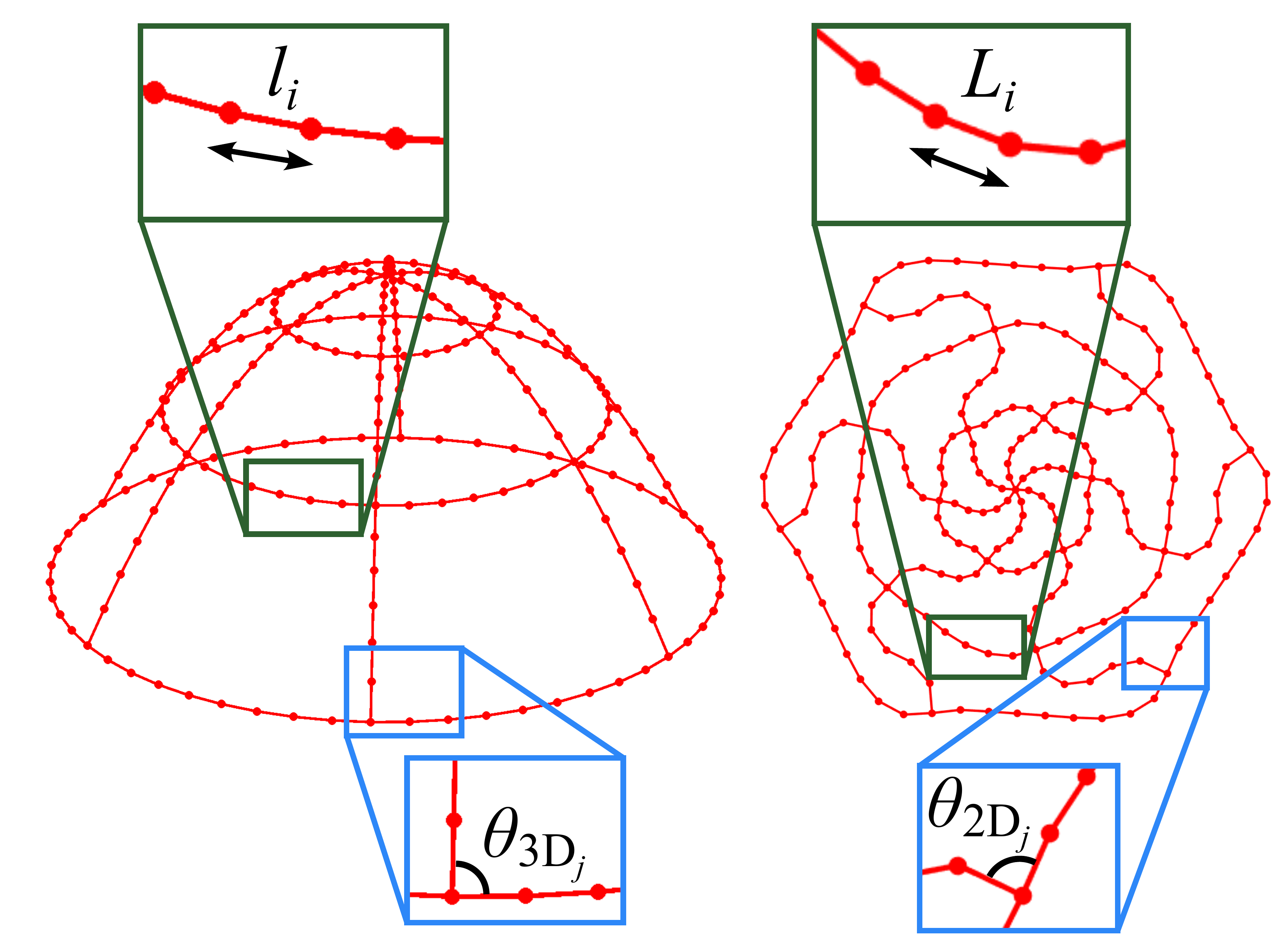}
    \caption{\textbf{An illustration of length-preserving and angle-preserving constraints.} As highlighted in the green boxes, the length-preserving constraints aim to ensure that the corresponding rods in the original 3D rod-based structure (left) and the 2D planar embedding (right) are equal in length, i.e., $l_i = L_i$ for all $i$. The blue boxes highlight the angle-preserving constraints, which aim to ensure that the corresponding angles at the major joints in 3D and 2D are equal, i.e., $\theta_{3D_j} = \theta_{2D_j}$.}
    \label{fig:illustration_length_angle}
\end{figure}

\subsubsection{Length-preserving constraints}
Specifically, for the length-preserving constraints, we aim to enforce the preservation of the length for all rod segments in the planar embedding. To achieve this, let $l_{i} = \|v_{e_i(1)} -v_{e_i(2)}\|$  to be the length of each rod segment $e_{i} = [e_i(1), e_i(2)]$  in the 3D rod-based structure $(\mathcal{V},\mathcal{E})$, and let $L_{i} = \|p_{e_i(1)} -p_{e_i(2)}\|$ be the length of the corresponding rod segment in the planar embedding $(\mathcal{P},\mathcal{E_{\text{2D}}})$. Here, $\|\cdot\|$ is defined to be the Euclidean norm. To preserve the length, it is desired to have $l_i = L_i$ (see Fig.~\ref{fig:illustration_length_angle} for an illustration). From a computational perspective, since different rods in the overall rod-based structures may have different lengths, to ensure a fair consideration of all rods, we have the following length-preserving constraint:
    \begin{equation} \label{eqt:constraint_length}
        E_{\text{L}}(e_i) = \frac{L_i}{l_i} - 1 = 0.
    \end{equation}
It is easy to see that $E_{\text{L}}(e_i) = 0$ if and only if $L_i = l_i$. Also, note that $E_{\text{L}}$ is dimensionless.

Now, we further derive the gradient of this constraint by differentiating it with respect to the vertices in the planar embedding $\mathcal{P}$. For each $i$, note that $l_i$ is given by the original 3D rod-based structure and can be treated as a constant. Let $p_{e_i(1)} = (x_{i_1}, y_{i_1})$ and $p_{e_i(2)} = (x_{i_2}, y_{i_2})$. We have
    \begin{equation}
    E_{\text{L}}(e_i) = \frac{\sqrt{(x_{i_1}-x_{i_2})^2+(y_{i_1}-y_{i_2})^2}}{l_i}-1.
    \end{equation}
Using the chain rule, we have
    \begin{equation}
        \frac{\partial E_\text{L}(e_i)}{\partial x_{i_1}} = \frac{2(x_{i_1}-x_{i_2})}{2\sqrt{(x_{i_1}-x_{i_2})^2+(y_{i_1}-y_{i_2})^2}\cdot l_i} = \frac{x_{i_1}-x_{i_2}}{L_i\cdot l_i}.
    \end{equation}
Similarly, we have
    \begin{equation}
    \resizebox{0.9\linewidth}{!}{$
        \frac{\partial E_\text{L}(e_i)}{\partial y_{i_1}} = \frac{y_{i_1}-y_{i_2}}{L_i\cdot l_i}, \quad
        \frac{\partial E_\text{L}(e_i)}{\partial x_{i_2}} = \frac{-(x_{i_1}-x_{i_2})}{L_i\cdot l_i}, \quad
        \frac{\partial E_\text{L}(e_i)}{\partial y_{i_2}} = \frac{-(y_{i_1}-y_{i_2})}{L_i\cdot l_i}.
        $}
    \end{equation}
Thus, if we consider all variables as $(x_1, x_2, \dots, x_{m}, y_1, y_2, \dots,  y_{m})$, where $m$ is the total number of vertices, we have
    \begin{equation} \label{eqt:gradient_length}
    \resizebox{1\linewidth}{!}{$
        \nabla E_\text{L}(e_i)  = \begin{bmatrix} 
        0 & \cdots & 0 & \frac{\partial E_\text{L}(e_i)}{\partial x_{i_1}} & 0 & \cdots & \frac{\partial E_\text{L}(e_i)}{\partial x_{i_2}} & 0 & \cdots & \frac{\partial E_\text{L}(e_i)}{\partial y_{i_1}}  & 0 & \cdots & \frac{\partial E_\text{L}(e_i)}{\partial y_{i_2}} &  \cdots & 0
        \end{bmatrix}^T.
        $}
    \end{equation}
In other words, for each $i$, there are exactly four non-zero entries in $\nabla E_\text{L}(e_i)$.
    
\subsubsection{Angle-preserving constraints}
Next, for the angle-preserving constraints, we aim to preserve the angles at the major joints of the rods, which are the vertices with vertex degree $\geq 3$. They correspond to the special location at which two chains of rods (each representing a curve) meet, instead of every intersection point between two adjacent rods. This ensures that the overall directions of the rod chains are largely preserved in the planar representation. From a more practical perspective, since such major joints are mechanically complex and are generally the more fragile part of the rod-based structures, preserving the angles at them would be beneficial for reducing the distortion induced locally at them during shape morphing or deformations (see Fig.~\ref{fig:illustration_length_angle} for an illustration).

First, we denote all angles at the major joints in both the original 3D rod-based structure and the 2D planar embedding as $\{\theta_{3D_j}\}_{j=1}^q$ and $\{\theta_{2D_j}\}_{j=1}^q$, where $q$ is the total number of angles, and every pair of angles $(\theta_{3D_j}, \theta_{2D_j})$ are the angles formed the corresponding vertices in 3D and 2D. The angles can be easily expressed in terms of the vertex coordinates as detailed in~\cite{choi2019programming}. Then, for every $j = 1, 2, \dots, q$, we consider the angle-preserving constraint 
    \begin{equation} \label{eqt:constraint_angle}
        E_{\text{A}}(j)= \cos(\theta_{2D_j})-\cos(\theta_{3D_j}) = 0.
    \end{equation}
Here, we use $\cos(\theta_{2D_j}), \cos(\theta_{3D_j})$ instead of $\theta_{2D_j}, \theta_{3D_j}$ so that the gradient of $E_{\text{A}}(j)$ can be expressed in terms of vertex coordinates more easily. Specifically, note that $\theta_{3D_j}$ is given by the original 3D rod-based structure and hence $\cos(\theta_{3D_j})$ can be treated as a constant. Now, for each angle $\theta_{2D_j}$, we look up the three corresponding vertices in the planar embedding. Denote them as
 $p_{\theta_j(1)} =(x_{j_1},y_{j_1})$, $p_{\theta_j(2)} =(x_{j_2},y_{j_2})$, and $p_{\theta_j(3)} =(x_{j_3},y_{j_3})$. We further let
   \begin{equation}
        \vec{a} = \begin{bmatrix}
            x_{j_2}\\y_{j_2}
        \end{bmatrix} 
        - \begin{bmatrix}
            x_{j_1}\\y_{j_1}
        \end{bmatrix} \ \ \ \text{ and } \ \ \ 
        \vec{b} = \begin{bmatrix}
            x_{j_3}\\y_{j_3}
        \end{bmatrix} 
        - \begin{bmatrix}
            x_{j_1}\\y_{j_1}
        \end{bmatrix}.
    \end{equation}
Then, we have
    \begin{equation}
        \cos(\theta_{2D}) = \frac{\vec{a} \cdot \vec{b}}{||\vec{a}|| \ ||\vec{b}||}.
    \end{equation}
By writing $c = \vec{a} \cdot \vec{b}$, we can take partial derivatives of $c$ with respect to $x_{j_1}$ and $y_{j_1}$:
    \begin{equation}
        \frac{\partial c}{\partial x_{j_1}} = 2x_{j_1} - x_{j_2} - x_{j_3},  \ \ \frac{\partial c}{\partial y_{j_1}} = 2y_{j_1} - y_{j_2} - y_{j_3}.
    \end{equation}
 Similarly, we can take partial derivatives of $c$ with respect to $x_{j_2}$, $y_{j_2}$, $x_{j_3}$, $y_{j_3}$ and get:
    \begin{equation}
    \resizebox{1\linewidth}{!}{$
        \frac{\partial c}{\partial x_{j_2}} = x_{j_3} - x_{j_1},  \quad  \frac{\partial c}{\partial y_{j_2}} = y_{j_3} - y_{j_1}, \quad  \frac{\partial c}{\partial x_{j_3}} = x_{j_2} - x_{j_1},   \quad  \frac{\partial c}{\partial y_{j_3}} = y_{j_2} - y_{j_1}.
        $}
    \end{equation}
Next, we take the partial derivatives of $||\vec{a}||$ with respect to $x_{j_1}$, $y_{j_1}$:
    \begin{equation}
        \frac{\partial ||\vec{a}||}{\partial x_{j_1}} = \frac{x_{j_1} - x_{j_2}}{||\vec{a}||}, \quad \frac{\partial ||\vec{a}||}{\partial y_{j_1}} = \frac{y_{j_1} - y_{j_2}}{||\vec{a}||}, 
    \end{equation}
Similarly, for $x_{j_2}$, $y_{j_2}$, we have
     \begin{equation} 
        \frac{\partial ||\vec{a}||}{\partial x_{j_2}} = \frac{x_{j_2} - x_{j_1}}{||\vec{a}||}, \quad
        \frac{\partial ||\vec{a}||}{\partial y_{j_2}} = \frac{y_{j_2} - y_{j_1}}{||\vec{a}||}.
    \end{equation}
Since $\vec{a}$ does not involve $x_{j_3}$, $y_{j_3}$, it is easy to see that
    \begin{equation}
        \frac{\partial ||\vec{a}||}{\partial x_{j_3}} = \frac{\partial ||\vec{a}||}{\partial y_{j_3}} = 0.
    \end{equation}
Analogously, for $ ||\vec{b}||$, we have:
    \begin{align}
        \frac{\partial ||\vec{b}||}{\partial x_{j_1}} = \frac{x_{j_1} - x_{j_3}}{||\vec{b}||}, &\quad
        \frac{\partial ||\vec{b}||}{\partial y_{j_1}} = \frac{y_{j_1} - y_{j_3}}{||\vec{b}||}, \\
        \frac{\partial ||\vec{b}||}{\partial x_{j_2}} &= \frac{\partial ||\vec{b}||}{\partial y_{j_2}} = 0, \\
        \frac{\partial ||\vec{b}||}{\partial x_{j_3}} = \frac{x_{j_3} - x_{j_1}}{||\vec{b}||}, &\quad
        \frac{\partial ||\vec{b}||}{\partial y_{j_3}} = \frac{y_{j_3} - y_{j_1}}{||\vec{b}||}.
    \end{align}
    Thus, for $k = {1,2,3}$, the partial derivative of the angle constraint $E_\text{A}(j)$ with respect to all $x_{j_k}, y_{j_k}$ are:
    \begin{equation}
        \frac{\partial E_\text{A}(j)}{\partial x_{j_k}} = \frac{\frac{\partial c}{\partial x_{j_k}}||\vec{a}||||\vec{b}||-(\vec{a} \cdot \vec{b})(\frac{\partial ||\vec{a}||}{\partial x_{j_k}}||\vec{b}|| + \frac{\partial ||\vec{b}||}{\partial x_{j_k}}||\vec{a}||)}{||\vec{a}||^2||\vec{b}||^2}
    \end{equation}
    and
        \begin{equation}
        \frac{\partial E_\text{A}(j)}{\partial y_{j_k}} = \frac{\frac{\partial c}{\partial y_{j_k}}||\vec{a}||||\vec{b}||-(\vec{a} \cdot \vec{b})(\frac{\partial ||\vec{a}||}{\partial y_{j_k}}||\vec{b}|| + \frac{\partial ||\vec{b}||}{\partial y_{j_k}}||\vec{a}||)}{||\vec{a}||^2||\vec{b}||^2},
    \end{equation}
    in which all the intermediate partial derivative terms have already been derived above.
 
Thus, if we consider all variables as $(x_1, x_2, \dots, x_{m}, y_1, y_2, \dots,  y_{m})$, we have
    \begin{equation} \label{eqt:gradient_angle}
    \resizebox{1\linewidth}{!}{$
        \nabla E_\text{A}(j)  = 
        \begin{bmatrix} 
        0 & \cdots  & \frac{\partial E_\text{A}(j)}{\partial x_{j_1}}  & \cdots & \frac{\partial E_\text{A}(j)}{\partial x_{j_2}}  & \cdots & \frac{\partial E_\text{A}(j)}{\partial x_{j_3}}  & \cdots & \frac{\partial E_\text{A}(j)}{\partial y_{j_1}}  & \cdots & \frac{\partial E_\text{A}(j)}{\partial y_{j_2}}  & \cdots & \frac{\partial E_\text{A}(j)}{\partial y_{j_3}} &  \cdots & 0
        \end{bmatrix}^T.
        $}
    \end{equation}
In other words, for each $j$, there are exactly six non-zero entries in $\nabla E_\text{A}(j)$.

\subsubsection{No-overlap constraints}
As discussed previously, the third type of constraints in our formulation is the no-overlap constraints, which aim to prevent the overlap between rod segments in the planar embedding result. While the most direct way is to consider all pairs of rod segments and check all possible occurrences of overlaps, the large number of combinations is computationally expensive. Moreover, designing a suitable constraint function for preventing such occurrences is highly nontrivial. Therefore, instead of directly working with the overlaps, we propose an alternative implicit formulation for the no-overlap constraints.  

More specifically, denote $\widetilde{\mathcal{B}}$ as the collection of all boundary vertices of the initial planar embedding. We first consider building a triangulation $\{\widetilde{T}_1, \widetilde{T}_2, \dots,\widetilde{T}_n\}$ on the initial planar embedding result so that the sum of all triangle areas is equal to the area enclosed by the vertices in $\widetilde{\mathcal{B}}$, i.e.,
\begin{equation}
    \sum_{i = 1}^n  \text{Area}(\widetilde{T}_i) = \text{Area}(\widetilde{\mathcal{B}}).
\end{equation}
Now, to enforce that there is no overlap in the shape optimization result, we consider the following no-overlap constraint:
\begin{equation} \label{eqt:constraint_overlap}
    E_\text{O} = \sum_{i = 1}^n  \text{Area}(T_i) - \text{Area}(\mathcal{B}) = 0,
\end{equation}
where $T_i$ represents the $i$-th triangle with the updated positions $p_{T_i(1)} = (x_{i_1},y_{i_1})$, $p_{T_i(2)} = (x_{i_2},y_{i_2})$ and $p_{T_i(3)} = (x_{i_3},y_{i_3})$ induced by the above triangulation, and $\mathcal{B}$ represents the updated coordinates of the vertices on the boundary induced by $\widetilde{\mathcal{B}}$. To justify this formulation, we establish the following result:
\begin{theorem}
If the area enclosed by all boundary vertices in the shape optimization result is equal to the area sum of all triangles, then there is no overlap in the shape optimization result.
\end{theorem}

\textit{Proof. } The proof is straightforward and hence omitted here. See Appendix for the details.
\hfill $\blacksquare$

From the above result, we see that enforcing Eq.~\eqref{eqt:constraint_overlap} will effectively prevent the occurrence of overlaps throughout the shape optimization process. However, we remark that the other direction of the statement does not hold. In Fig.~\ref{fig:illustration_triangulation}, we show a simple counterexample. In particular, Fig.~\ref{fig:illustration_triangulation}(a) shows an example of the initial embedding of a rod-based structure and the Delaunay triangulation constructed based on all its vertices. In Fig.~\ref{fig:illustration_triangulation}(b), we show a deformed configuration of the planar embedding. It is noteworthy that the deformed rod-based structure does not contain any rod overlaps. However, several overlaps can be found in the underlying triangulation under the deformation, indicating that the area enclosed by all boundary vertices in the deformed triangulation is not equal to the sum of all individual triangle areas. This example shows that the other direction of the theorem does not necessarily hold. In other words, while the condition can fulfill our needs for preventing overlaps between rods, it may be too strict in some cases and overconstrains our problem. In a later section, we will discuss how we can address this aspect by introducing an alternating minimization procedure in our overall proposed algorithm.

\begin{figure}[t!]
    \centering
    \includegraphics[width=0.8\textwidth]{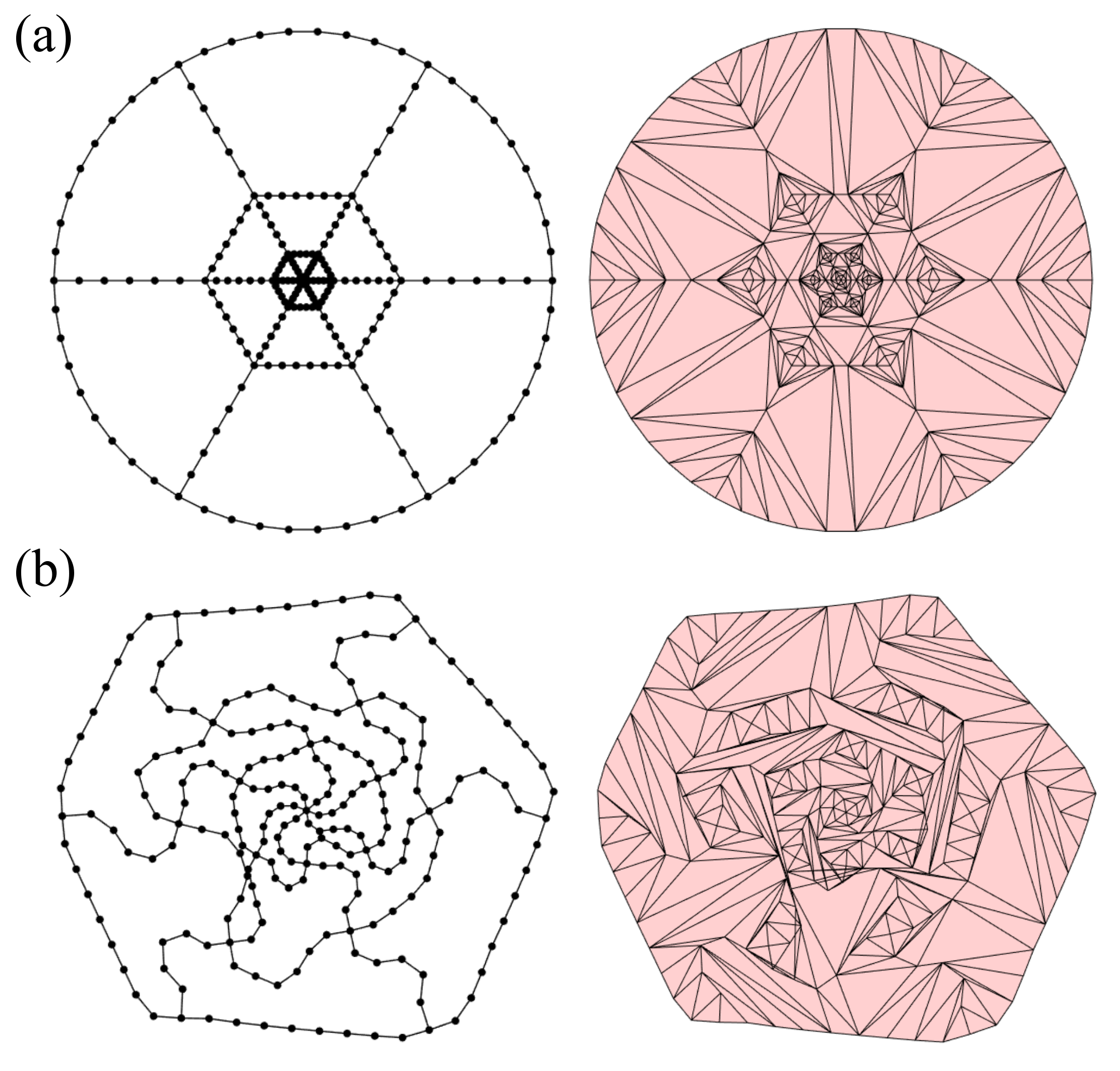}
    \caption{\textbf{An illustration of the construction of triangulations for the no-overlap constraint.} (a) An initial embedding of a rod-based structure and a Delaunay triangulation constructed on the set of all nodes. (b) An example of a deformed configuration of the planar embedding, with the associated triangulation induced from (a). It can be observed that the deformed rod-based structure does not contain overlaps, while several mesh overlaps can be found in the induced deformed triangulation, suggesting that the no-overlap constraint in Eq.~\eqref{eqt:constraint_overlap} is sufficient but not necessary for ensuring the no-overlap condition.}
    \label{fig:illustration_triangulation}
\end{figure}
     
Now, we further derive the gradient of the constraint $E_\text{O}$ in Eq.~\eqref{eqt:constraint_overlap} as follows. First, we express $E_\text{O} = H - S$ in terms of the vertices in $\mathcal{P}$, where $H$ represents the triangle area sum term and $S$ represents the boundary area term. We handle the two terms one by one. First, for each triangle $T_i$, we can let $p_{T_i(1)} = (x_{i_1},y_{i_1})$, $p_{T_i(2)} = (x_{i_2},y_{i_2})$ and $p_{T_i(3)} = (x_{i_3},y_{i_3})$. If we denote the length of the three edges of $T_i$ as $L_{12_i}$, $L_{23_i}$, $L_{31_i}$, we have
\begin{equation}
    L_{12_i} = \sqrt{(x_{i_1}-x_{i_2})^2+(y_{i_1}-y_{i_2})^2},
\end{equation} 
\begin{equation}
    L_{23_i} = \sqrt{(x_{i_2}-x_{i_3})^2+(y_{i_2}-y_{i_3})^2},
\end{equation} 
\begin{equation}
    L_{31_i} = \sqrt{(x_{i_3}-x_{i_1})^2+(y_{i_3}-y_{i_1})^2}.
\end{equation} 
Therefore, using the Heron's formula, the triangle area sum is given by
\begin{equation} \label{eqt:heron}
    H = \sum_{i=1}^n\sqrt{s_i(s_i-L_{12_i})(s_i-L_{23_i})(s_i-L_{31_i})},
\end{equation}
where \begin{equation}
    s_i = \frac{L_{12_i}+L_{23_i}+L_{31_i}}{2}
\end{equation}
for all $i = 1, 2, \dots, n$.

As for the boundary area term, if we denote the indices of the boundary vertices as $b_1, b_2, \dots, b_s$, using the shoelace formula, the boundary area is given by  
\begin{equation} \label{eqt:shoelace}
    S = \frac{1}{2}(x_{b_1}y_{b_2}+x_{b_2}y_{b_3}+\dots+x_{b_s}y_{b_1} - x_{b_2}y_{b_1}-x_{b_3}y_{b_2}-\dots-x_{b_1}y_{b_s}).
\end{equation}
From Eq.~\eqref{eqt:heron} and Eq.~\eqref{eqt:shoelace}, we see that $E_\text{O}$ can be expressed explicitly in terms of the vertices in $\mathcal{P}$. 

To get the gradient of $E_\text{O}$, we have to differentiate it with respect to the whole $\mathcal{P}$. For the triangulation part $H$, note that each $T_i$ only involves the vertex coordinates $x_{i_1}$, $x_{i_2}$, $x_{i_3}$ and $y_{i_1}$, $y_{i_2}$, $y_{i_3}$, and hence the derivative of the area term for $T_i$ with respect to all other variables can be omitted. Below, for convenience, we write $x_{i_1} = x_1$ and similarly for $x_2$, $x_3$, $y_1$,  $y_2$,  $y_3$ to illustrate the derivation once. For $\alpha$, $\beta$, $\gamma = 1,2,3$, where $\alpha \neq \beta \neq \gamma$, we can easily get
\begin{align}
    \frac{\partial L_{\alpha \beta}}{\partial x_{\alpha}} = \frac{x_\alpha - x_\beta}{L_{\alpha\beta}}, &\quad
    \frac{\partial L_{\alpha \beta}}{\partial y_{\alpha}} = \frac{y_\alpha - y_\beta}{L_{\alpha\beta}},\\
    \frac{\partial L_{\alpha \beta}}{\partial x_{\beta}} = \frac{-(x_\alpha - x_\beta)}{L_{\alpha\beta}},  &\quad  \frac{\partial L_{\alpha \beta}}{\partial y_{\beta}} = \frac{-(y_\alpha - y_\beta)}{L_{\alpha\beta}}, \\
    \frac{\partial L_{\alpha\beta}}{\partial x_\gamma} &= \frac{\partial L_{\alpha\beta}}{\partial y_\gamma} = 0.
\end{align}
Then, for $m=1,2,3$, we have
\begin{equation}
    \resizebox{0.9\linewidth}{!}{$
    \frac{\partial \text{Area}(T_i)}{\partial x_m} = \frac{4L_{12}(L_{23}^2+L_{31}^2-L_{12}^2)\frac{\partial L_{12}}{\partial x_k}+4L_{23}(L_{31}^2+L_{12}^2-L_{23}^2)\frac{\partial L_{23}}{\partial x_k}+4L_{31}(L_{12}^2+L_{23}^2-L_{31}^2)\frac{\partial L_{31}}{\partial x_k}}{8\sqrt{(L_{12}+L_{23}+L_{31})(L_{12}+L_{23}-L_{31})(L_{12}-L_{23}+L_{31})(-L_{12}+L_{23}+L_{31})}}
    $}
\end{equation}
and
\begin{equation}
    \resizebox{0.9\linewidth}{!}{$
    \frac{\partial \text{Area}(T_i)}{\partial y_m} = \frac{4L_{12}(L_{23}^2+L_{31}^2-L_{12}^2)\frac{\partial L_{12}}{\partial y_k}+4L_{23}(L_{31}^2+L_{12}^2-L_{23}^2)\frac{\partial L_{23}}{\partial y_k}+4L_{31}(L_{12}^2+L_{23}^2-L_{31}^2)\frac{\partial L_{31}}{\partial y_k}}{8\sqrt{(L_{12}+L_{23}+L_{31})(L_{12}+L_{23}-L_{31})(L_{12}-L_{23}+L_{31})(-L_{12}+L_{23}+L_{31})}}.
    $}
\end{equation}
Hence, for each $k = 1,2,3,\dots,m$, we simply need to add up the partial derivatives of all individual triangle area terms. We have:
\begin{equation}
    \frac{\partial H}{\partial x_k} = \sum_{i=1}^n\frac{\partial \text{Area}(T_i)}{\partial x_k} \quad \text{and} \quad
    \frac{\partial H}{\partial y_k} = \sum_{i=1}^n\frac{\partial \text{Area}(T_i)}{\partial y_k}.
\end{equation}
This shows how the partial derivatives of $H$ with respect to each vertex coordinate can be derived.

For the boundary area term $S$, it suffices to differentiate it with respect to the coordinates of all boundary vertices $\mathcal{B}$, i.e., $x_{b_1},\space x_{b_2},\space x_{b_3},\dots,\space x_{b_s}$ and $y_{b_1},\space y_{b_2},\space y_{b_3},\dots,\space y_{b_s}$. For $k = 1, \dots,s$, from Eq.~\eqref{eqt:shoelace} we can easily get
\begin{equation}
    \frac{\partial S}{\partial x_{b_k}} = \frac{y_{b_{(k+1)}} - y_{b_{(s+1-k)}}}{2} \quad \text{and} \quad
    \frac{\partial S}{\partial y_{b_k}} = \frac{x_{b_{(s+1-k)}} - x_{b_{(k+1)}}}{2}.
\end{equation}
For all other vertices $p_k\in \mathcal{P} \setminus \mathcal{B}$, we have $\frac{\partial S}{\partial x_k} = \frac{\partial S}{\partial y_k} = 0$.

Putting the above results together, the gradient of the no-overlap constraint is given by
\begin{equation} \label{eqt:gradient_overlap}
    \resizebox{0.25\linewidth}{!}{$
        \nabla E_\text{O} = \begin{bmatrix}
            \partial_{x_1}{H} - \partial_{x_1}{S}\\
            \partial_{x_2}{H} - \partial_{x_2}{S}\\
            \partial_{x_3}{H} - \partial_{x_3}{S}\\
            \vdots\\
            \partial_{x_m}{H} - \partial_{x_m}{S}\\
            \partial_{y_1}{H} - \partial_{y_1}{S}\\
            \partial_{y_2}{H} - \partial_{y_2}{S}\\
            \partial_{y_3}{H} - \partial_{y_3}{S}\\
            \vdots\\
            \partial_{y_m}{H} - \partial_{y_m}{S}\\
            \end{bmatrix}.
            $}
    \end{equation}

\subsubsection{The objective function and the constrained optimization formulation}

After describing all constraints to be satisfied, we move on to the design of the objective function in our shape optimization problem. Note that in the constraints, we have primarily focused on the length of each individual rod segment and the angles at the major joints. To further achieve a low overall geometric distortion of the structure, we may consider the remaining angle quantities in the structure that we have not yet covered, i.e., the angles at the remaining joints between adjacent individual rod segments. It is easy to see that if the values of these angles in the planar embedding are largely similar to those in the 3D structure, then it implies that there are fewer unnatural bends or sharp orientation changes among neighboring rods, and hence the overall geometric distortion of the embedding should be small.

Motivated by the above, we consider the following angle-based objective function: 
\begin{equation} \label{eqt:objective_function}
    E = \sum_{i=1}^r \left(\cos(\phi_{2D_{i}}) - \cos(\phi_{3D_{i}}) \right)^2,
\end{equation}
where $r$ is the total number of angles between adjacent rod segments, $\phi_{2D_{i}}$ is the angle value in the planar embedding, and $\phi_{3D_{i}}$ is the corresponding angle value in the original 3D structure. To simplify our discussion below, we further denote 
\begin{equation}
    E_i = \left(\cos(\phi_{2D_{i}}) - \cos(\phi_{3D_{i}}) \right)^2.
\end{equation}

Analogous to the discussion of the angle-preserving constraint in Eq.~\eqref{eqt:constraint_angle}, here note that we can express $E$ and its gradient in terms of the 2D vertex coordinates. Again, all $\phi_{3D_{i}}$ values are given by the original 3D structure and hence $\cos(\phi_{3D_{i}})$ can be treated as a constraint. For the term $\cos(\phi_{2D_{i}})$, note that each angle involves three vertices and we denote them as $p_{\phi_{2D_i}(1)} = (x_1,y_1)$,  $p_{\phi_{2D_i}(2)} = (x_2,y_2)$, $p_{\phi_{2D_i}(3)} = (x_3,y_3)$. Then we have
\begin{equation}
    \resizebox{0.7\linewidth}{!}{$
        \vec{a} = \begin{bmatrix}
            x_2\\y_2
        \end{bmatrix} 
        - \begin{bmatrix}
            x_1\\y_1
        \end{bmatrix}, \quad
        \vec{b} = \begin{bmatrix}
            x_3\\y_3
        \end{bmatrix} 
        - \begin{bmatrix}
            x_1\\y_1
        \end{bmatrix}, \quad
        \cos(\phi_{2D_i}) = \displaystyle \frac{\vec{a} \cdot \vec{b}}{||\vec{a}||\space||\vec{b}||}.
        $}
    \end{equation}
  Letting $c = \vec{a} \cdot \vec{b}$, we get
\begin{align}
        \frac{\partial c}{\partial x_1} = 2x_1 - x_2 - x_3, & \quad
        \frac{\partial c}{\partial y_1} = 2y_1 - y_2 - y_3, \\
       \frac{\partial c}{\partial x_2} = x_3 - x_1, &\quad
        \frac{\partial c}{\partial y_2} = y_3 - y_1, \\
        \frac{\partial c}{\partial x_3} = x_2 - x_1, &\quad
        \frac{\partial c}{\partial y_3} = y_2 - y_1.
    \end{align}
    Also, 
\begin{align}
        \frac{\partial ||\vec{a}||}{\partial x_1} = \frac{x_1 - x_2}{||\vec{a}||}, & \quad
        \frac{\partial ||\vec{a}||}{\partial y_1} = \frac{y_1 - y_2}{||\vec{a}||}, \\
        \frac{\partial ||\vec{a}||}{\partial x_2} = \frac{x_2 - x_1}{||\vec{a}||}, & \quad
        \frac{\partial ||\vec{a}||}{\partial y_2} = \frac{y_2 - y_1}{||\vec{a}||}, \\
        \frac{\partial ||\vec{a}||}{\partial x_3} & = \frac{\partial ||\vec{a}||}{\partial y_3} = 0,
    \end{align}
    and
    \begin{align}
        \frac{\partial ||\vec{b}||}{\partial x_1} = \frac{x_1 - x_3}{||\vec{b}||}, &\quad
        \frac{\partial ||\vec{b}||}{\partial y_1} = \frac{y_1 - y_3}{||\vec{b}||}, \\
        \frac{\partial ||\vec{b}||}{\partial x_2} &= \frac{\partial ||\vec{b}||}{\partial y_2} = 0, \\
        \frac{\partial ||\vec{b}||}{\partial x_3} = \frac{x_3 - x_1}{||\vec{b}||}, &\quad
        \frac{\partial ||\vec{b}||}{\partial y_3} = \frac{y_3 - y_1}{||\vec{b}||}.
    \end{align}
From the above, for each term $E_i$ and $k = 1,2,3$, we have
\begin{equation}
    \resizebox{0.9\linewidth}{!}{$
        \frac{\partial E_i}{\partial x_k} = \frac{2\left(\frac{\vec{a} \cdot \vec{b}}{||\vec{a}||||\vec{b}||}-\cos(\phi_{3D_{i}})\right) \left[\frac{\partial c}{\partial x_k}||\vec{a}||||\vec{b}||-(\vec{a} \cdot \vec{b})(\frac{\partial ||\vec{a}||}{\partial x_k}||\vec{b}|| + \frac{\partial ||\vec{b}||}{\partial x_k}||\vec{a}||)\right]}{||\vec{a}||^2||\vec{b}||^2},
        $}
\end{equation}
\begin{equation}
    \resizebox{0.9\linewidth}{!}{$
        \frac{\partial E_i}{\partial y_k} = \frac{2\left(\frac{\vec{a} \cdot \vec{b}}{||\vec{a}||||\vec{b}||}-\cos(\phi_{3D_{i}})\right) \left[ \frac{\partial c}{\partial y_k}||\vec{a}||||\vec{b}||-(\vec{a} \cdot \vec{b})(\frac{\partial ||\vec{a}||}{\partial y_k}||\vec{b}|| + \frac{\partial ||\vec{b}||}{\partial y_k}||\vec{a}||)\right]}{||\vec{a}||^2||\vec{b}||^2}.
        $}
\end{equation}
Therefore, for each $p_k = (x_k,y_k) \in \mathcal{P}$, where $k= 1,2,3,\dots,m$, we have
\begin{equation}
    \frac{\partial E}{\partial x_k} = \sum_{i=1}^r\frac{\partial E_i}{\partial x_k} \quad \text{and} \quad
    \frac{\partial E}{\partial y_k} = \sum_{i=1}^r\frac{\partial E_i}{\partial y_k}.
\end{equation}
Hence, the gradient of the objective function is given by
\begin{equation} \label{eqt:objective_gradient}
    \resizebox{0.9\linewidth}{!}{$
    \nabla E = \begin{bmatrix}
        \partial_{x_1}E & 
        \partial_{x_2}E &
        \partial_{x_3}E &
        \cdots& 
        \partial_{x_m}E& 
        \partial_{y_1}E& 
        \partial_{y_2}E& 
        \partial_{y_3}E& 
        \cdots&
        \partial_{y_m}E
    \end{bmatrix}^T.
        $}
\end{equation}

Altogether, in our constrained optimization problem, we minimize the objective function in Eq.~\eqref{eqt:objective_function} subject to the length-preserving constraints in Eq.~\eqref{eqt:constraint_length}, the angle-preserving constraints in Eq.~\eqref{eqt:constraint_angle}, and the no-overlap constraints in Eq.~\eqref{eqt:constraint_overlap}. For both the objective function and constraints, we further utilize their gradients as derived in Eq.~\eqref{eqt:gradient_length}, Eq.~\eqref{eqt:gradient_angle}, Eq.~\eqref{eqt:gradient_overlap}, and Eq.~\eqref{eqt:objective_gradient} in the numerical optimization process.

\subsection{Overlap correction}
As mentioned previously, while the no-overlap constraint will attempt to prevent the occurrence of overlaps, it may overconstrain the problem. In other words, in some situations, it may happen that the planar embedding result is already overlap-free but the no-overlap constraint value is still non-zero. This may affect the computational efficiency of our optimization procedure as we may need to execute a large number of iterations or even need to wait until reaching some prescribed maximum number of iterations. Therefore, one strategy that we propose is that we can consider including or excluding the no-overlap constraint in our actual optimization procedure in an alternating manner. To achieve this, we need some extra procedures to detect and correct the overlaps in the intermediate embedding results. A correction scheme is developed as follows.

\begin{figure}[t]
    \centering
    \includegraphics[width=\linewidth]{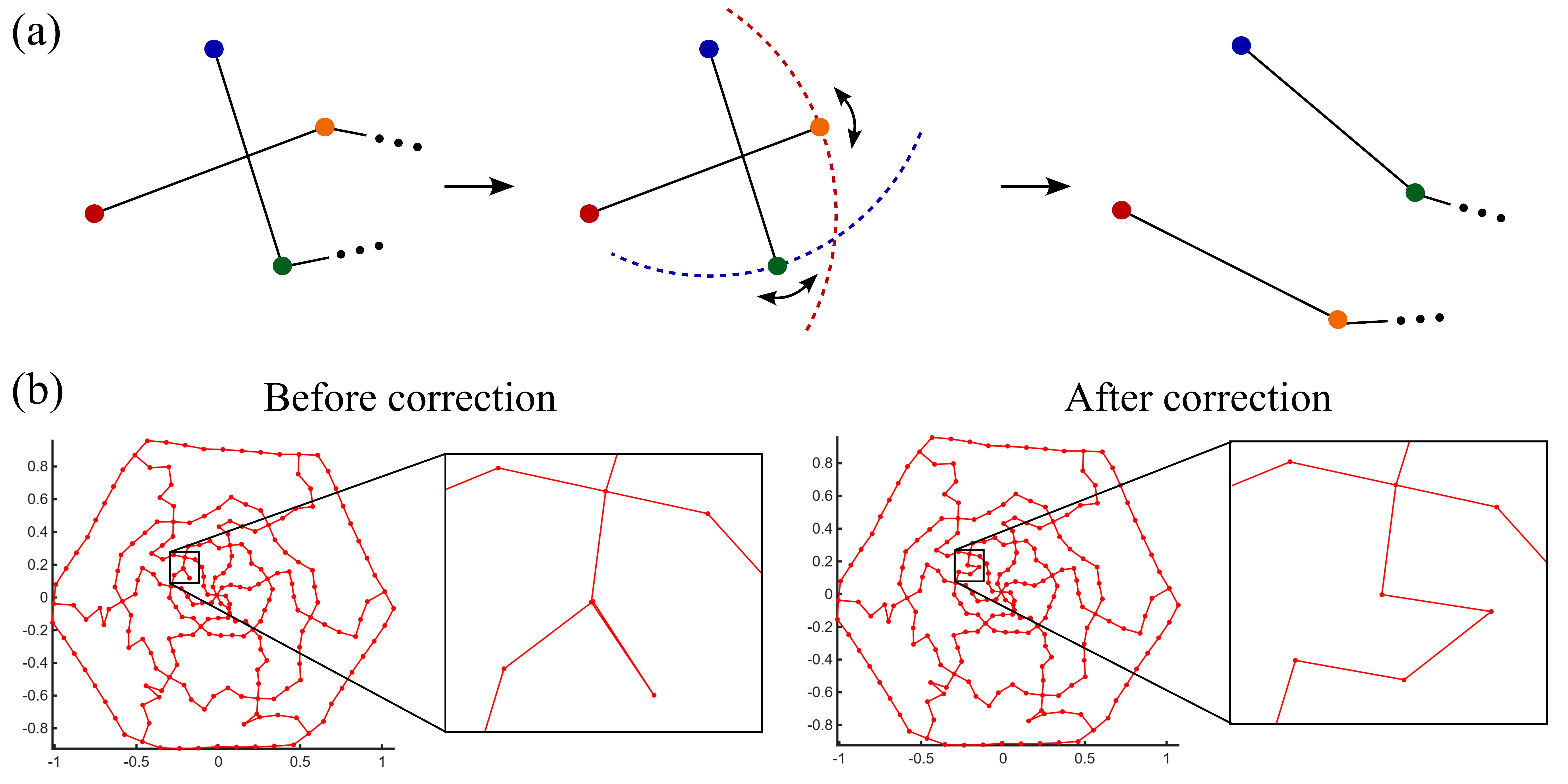}
    \caption{\textbf{An illustration of the overlap correction scheme.} (a) For each overlap, we first identify the four relevant vertices (each with a distinct color). We then fix two of them as circle centers and adjust the positions of the other two vertices along the two circles (dotted lines) to resolve the overlap. (b) An example of applying this overlap correction scheme to the planar embedding of rod-based structures. The left panel shows the configuration before the correction, with a zoom-in of the overlap position. The right panel shows the configuration after the correction.}
    \label{fig:illustration_correction}
\end{figure}

Between every iteration of each round of optimization run (with the no-overlap constraint included or excluded), we add a new process to remove any overlaps in the current configuration (see Fig.~\ref{fig:illustration_correction} for an illustration). First, we will detect all overlaps $O_n$ in the configuration. Note that each overlap $O_n$ will involve exactly two rod segments, which are associated with exactly four points. We denote them as ${\{P_{n_1},P_{n_2},P_{n_3},P_{n_4}\}}$. After getting this information, we use the Dijkstra algorithm to find the shortest path in the graph of the rod-based structure that will contain the four points. Also, we use the overlap point $O_n$ as the center and create a circle with radius $r_n={\max\{\|O_n - P_{n_i} \|, i = 1,2,3,4\}}$. With the circle we just generated, we can consider all rod segments for which the rod midpoints are inside or lying on the circle. Then, we keep every point on the computed shortest path unchanged except the points $P_{n_2}$ and $P_{n_3}$, Specifically, we allow each of those two points to shift along two circles with the center $P_{n_1}$ and $P_{n_4}$ and radii $\|P_{n_2}-P_{n_1}\|$ and $||P_{n_3}-P_{n_4}||$. We find two possible points that result in no overlaps inside the circle and in the path. In other words, we keep some key rod lengths unchanged while allowing the points to move the plane within a limited path in order to resolve the overlaps. Among all possible choices, we choose the point that gives the smallest length error with the vertex next to it in the path. If we cannot find any point that gives no overlaps, we will sacrifice some more length constraints and further consider two angle bisectors constructed by $P_{n_2}$, $O_n$ and $P_{n_3}$, $O_n$ respectively. We then find the possible points and choose the best one similarly as above.

We remark that in the above overlap correction procedure, we only control the two rod lengths $\|P_{n_2}-P_{n_1}\|$ and $||P_{n_3}-P_{n_4}||$ while the length of other neighboring rods may not be preserved. Also, as we do not impose any constraints on the angles, the angles between the rods may not be preserved. Moreover, since resolving each overlap involves the four associated vertices, if multiple overlap points $O_n$ are close to each other, a single step described above may only be able to resolve some overlaps but not all of them. Specifically, note that the flexibility in shifting the points along the circles as described above ensures that we can find a non-overlapping configuration locally at ${\{P_{n_1},P_{n_2},P_{n_3},P_{n_4}\}}$. However, there is no further guarantee or control over whether resolving one overlap will lead to overlaps elsewhere. Therefore, in practice, we set an iteration limit $N_{\max}$ as a fallback before going into the next iteration of the optimization problem, where $N_{\max} = 10$ in our experiments. Note that $N_{\max}$ can also be changed to a more adaptive threshold. For instance, since it is natural that the occurrence of overlaps may scale with the resolution of the structure, we may also set $N_{\max}$ to be dependent on the number of edges $|\mathcal{E}|$. Also, to prevent this process from increasing the number of overlaps, if the number of overlaps increases after an iteration of this process, we will break the correction process and directly move to the next optimization iteration. Ultimately, this extra overlap correction procedure aims to provide an improved initial guess for the subsequent optimization run, which will further optimize the vertex positions to reduce constraint violations. The overlap correction algorithm is summarized in Algorithm~\ref{alg:overlap}.

{\footnotesize
\begin{algorithm}[H]
    \KwData{A planar embedding $f:\mathcal{S}\rightarrow\mathbb{R}^2$, maximum iteration number $N_{\max}$.}
    \KwResult{A non-overlap transformation $T:\mathbb{R}^2\rightarrow\mathbb{R}^2$.}
    Detect all overlaps in the current configuration, and let $N_{\text{overlap}}$ be the number of overlaps\;
    Set $N_{\text{iter}} = 0$\;
    \While{$N_{\text{overlap}}>0$ and $N_{\text{iter}}<N_{\max}$}{
    \For{$i = 1, 2, \dots, N_{\text{overlap}}$}{
    Locate the four vertices that correspond to the $i$-th overlap\;
    Construct a circle with radius equal to the maximum length from the overlap to the 4 points\;
    Locate all other points from the planar embedding that lie on or inside the circle\;
    Construct two circles using the first and the fourth point as centers, with radii being the length between the second and the first point, and the length between the third and the fourth point, respectively\;
    Shift the second and third points along the two circles and find the optimal location that achieves the lowest length error and yields no overlap\;
    \If {Cannot find any point satisfying our desired result}{
     Consider the angle bisectors constructed by the second point and the overlap point, and the third point and the overlap point\;
     Allow the second and third points to shift on the angle bisectors again\;
    }}
    Update $N_{\text{overlap}}$ by the number of overlaps in the new configuration. Break the process if the number increases\;
    Set $N_{\text{iter}} = N_{\text{iter}}+1$\;
    }
    \caption{Overlap correction}
    \label{alg:overlap}
\end{algorithm}
}

\subsection{Summary}

After describing our initial embedding method, shape optimization formulation, and overlap correction procedure, we are now ready to formulate our overall framework for the low-distortion planar embedding of rod-based structures (see Fig.~\ref{fig:flowchart} for an overview).

\begin{figure}[t!]
    \centering
    \includegraphics[width=0.95\linewidth]{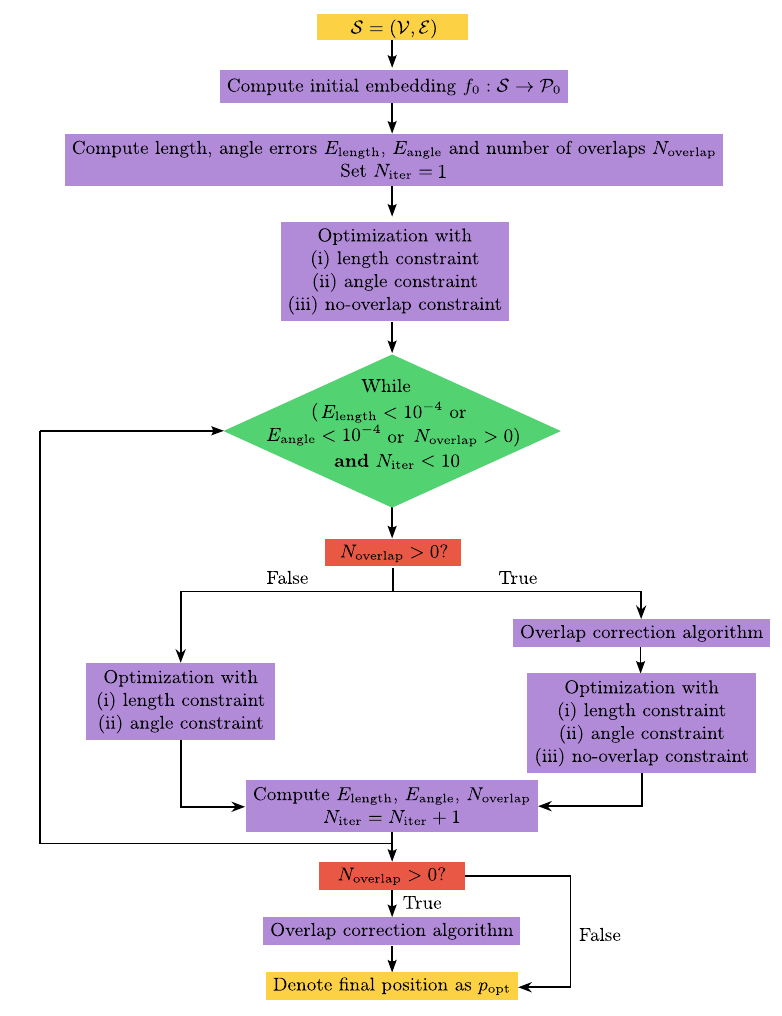}
    \caption{\textbf{A flowchart illustrating our proposed framework for low-distortion planar embedding of rod-based structures.}}
    \label{fig:flowchart}
\end{figure}

Specifically, given a 3D rod-based structure, we first apply the initial embedding method to obtain a 2D representation. We then solved our proposed constrained optimization problem to update the vertex positions on the plane for reducing the geometric distortion while preserving bijectivity in an iterative manner. In particular, for the consideration of computational efficiency, we first check and see whether the current configuration contains overlap. If it does not contain any overlaps, then at the current step we only solve the constrained optimization problem with the no-overlap constraint in Eq.~\eqref{eqt:constraint_overlap} skipped, so that we can focus on the reduction of the geometric distortion at this step without being overconstrained by the no-overlap constraint. Otherwise, we first apply the overlap correction scheme in Algorithm~\ref{alg:overlap} and then run the full constrained optimization procedure with the no-overlap constraints, with the goal of handling both the geometric distortion and bijectivity issue at this step. The above procedure is repeated until the prescribed requirements on both the geometric distortion and bijectivity are satisfied. In practice, the convergence can be determined by some prescribed threshold on the mean absolute length and angle errors  $E_{\text{length}}, E_{\text{angle}}$ and the maximum number of iterations. The proposed algorithm is summarized in Algorithm~\ref{alg:main}. 

{\footnotesize
\begin{algorithm}[H]
 \KwData{A 3D rod-based structure $\mathcal{S} = (\mathcal{V}, \mathcal{E})$.}
 \KwResult{A low-distortion planar embedding $f: \mathcal{S} \to \mathcal{P}$.}
Compute an initial embedding $f_0: \mathcal{S} \to \mathcal{P}_0$\;
Compute the length and angle quantities\;
Compute the length and angle errors $E_{\text{length}}, E_{\text{angle}}$\;
Detect all overlaps in the current configuration, and let $N_{\text{overlap}}$ be the number of overlaps\; 
Set $N_{\text{iter}} = 1$\;
Solve the full constrained optimization problem with the length, angle, and no-overlap constraints in Eq.~\eqref{eqt:constraint_length}, Eq.~\eqref{eqt:constraint_angle}, and Eq.~\eqref{eqt:constraint_overlap}\;
\While{($E_{\text{length}}> 10^{-4}$ or $E_{\text{angle}}> 10^{-4}$ or $N_{\text{overlap}}> 0$) and $N_{\text{iter}} <  10$}{
\uIf{$N_{\text{overlap}}=0$}{
Solve the constrained optimization problem without the no-overlap constraint in Eq.~\eqref{eqt:constraint_overlap}\;
}
\Else{
Apply the overlap correction algorithm (Algorithm~\ref{alg:overlap})\;
Solve the full constrained optimization problem with the length, angle, and no-overlap constraints in Eq.~\eqref{eqt:constraint_length}, Eq.~\eqref{eqt:constraint_angle}, and Eq.~\eqref{eqt:constraint_overlap}\;
}
Set $N_{\text{iter}} = N_{\text{iter}} + 1$\;
}

\If{$N_{\text{overlap}}>0$}{
Apply the overlap correction algorithm (Algorithm~\ref{alg:overlap})\;
}

Denote the final vertex positions as $p_{\text{opt}}$ and the shape optimization mapping as $g: \mathcal{P}_0 \to \mathcal{P}$. The final embedding $f$ is given by $f = g \circ f_0$, and we have $f(v_i) = (p_{\text{opt}})_i$\;
 \caption{Proposed framework for low-distortion planar embedding of rod-based structures}
 \label{alg:main}
\end{algorithm}
}

\section{Experimental results} \label{sect:experiment}

The proposed algorithms are implemented using MATLAB R2026a. All experiments are performed on a desktop computer with an Intel(R) Core(TM) i9-12900 2.40 GHz processor and 32 GB RAM. The optimization is solved using the interior-point method in IPOPT~\cite{wachter2006implementation}.

\subsection{Experiments with different rod-based structures}

We first quantify the geometric distortion of the planar embedding produced by our method. Here, the \emph{average length error} is defined as the mean of the absolute length error of each rod. The \emph{angle error} is defined as the mean of the absolute error of each angle that connects to more than two rods. Also, the number of overlaps in the planar embedding result is checked using the InterX function~\cite{interx} in MATLAB.

More specifically, for the illustration example shown in Fig.~\ref{fig:illustration}, we measure the error in length and angle for the planar embedding result and visualize the errors in Fig.~\ref{fig:errors}. Here, the rod segment length error for each rod is defined as the absolute difference between the original rod segment length in 3D $(L_i)$ and that in the resulting 2D representation $(l_i)$, i.e., $|l_i-L_i|$. As shown in Fig.~\ref{fig:errors}(a), the length error is of the order $10^{-16}$ for most rod segments, suggesting that the proposed method is highly length-preserving. We also assess the angle error at the major joints of the rod-based structure, defined as the absolute difference between the original angles in 3D $(\theta_{3D_j})$ and the corresponding angles in the final planar representation $(\theta_{2D_j})$, i.e., $|\theta_{2D_j}-\theta_{3D_j}|$. As shown in Fig.~\ref{fig:errors}(b), the angle error is very close to 0 at all major joints. This shows that the proposed method is highly accurate and is capable of preserving the required geometrical properties. Besides, one can also calculate the number of overlaps between all pairs of rod segments. The result shows that the planar embedding does not contain any overlaps. In other words, the bijectivity of the rod-based structure is well-preserved. In Fig.~\ref{fig:energy}, we further show the energy plot for this example. Here, it can be observed that the energy $E$ in Eq.~\eqref{eqt:objective_function} decreases rapidly throughout the iterations. Also, from the visualization of the intermediate configurations at different iteration steps, one can see that the result converges rapidly. This demonstrates the effectiveness of our proposed method.

\begin{figure}[t]
    \centering
    \includegraphics[width=\linewidth]{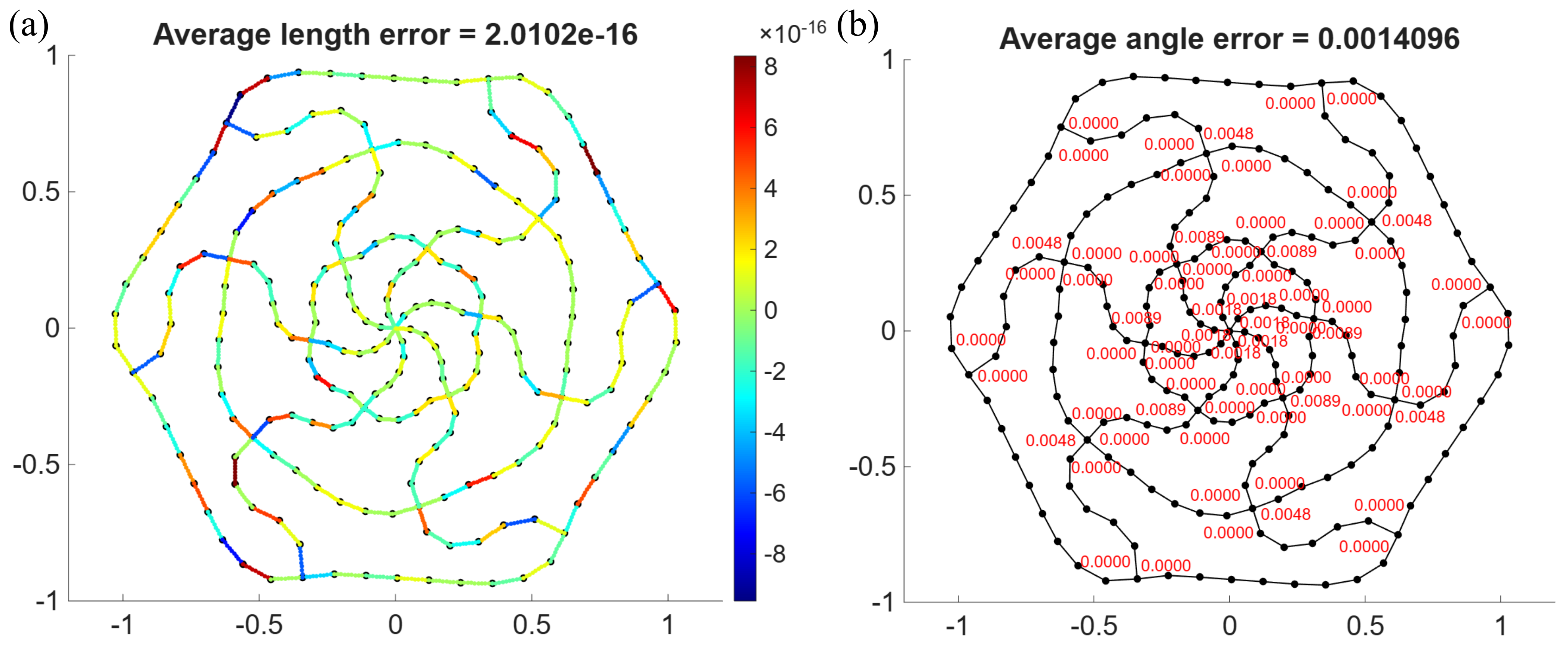}
    \caption{\textbf{Quantification of the geometric distortion of the example shown in Fig.~\ref{fig:illustration}.} (a)~The planar embedding result, with every rod segment color-coded with the length error. (b)~The angle error at each intersection point of the rods.}
    \label{fig:errors}
\end{figure}

\begin{figure}[t]
    \centering
    \includegraphics[width=\linewidth]{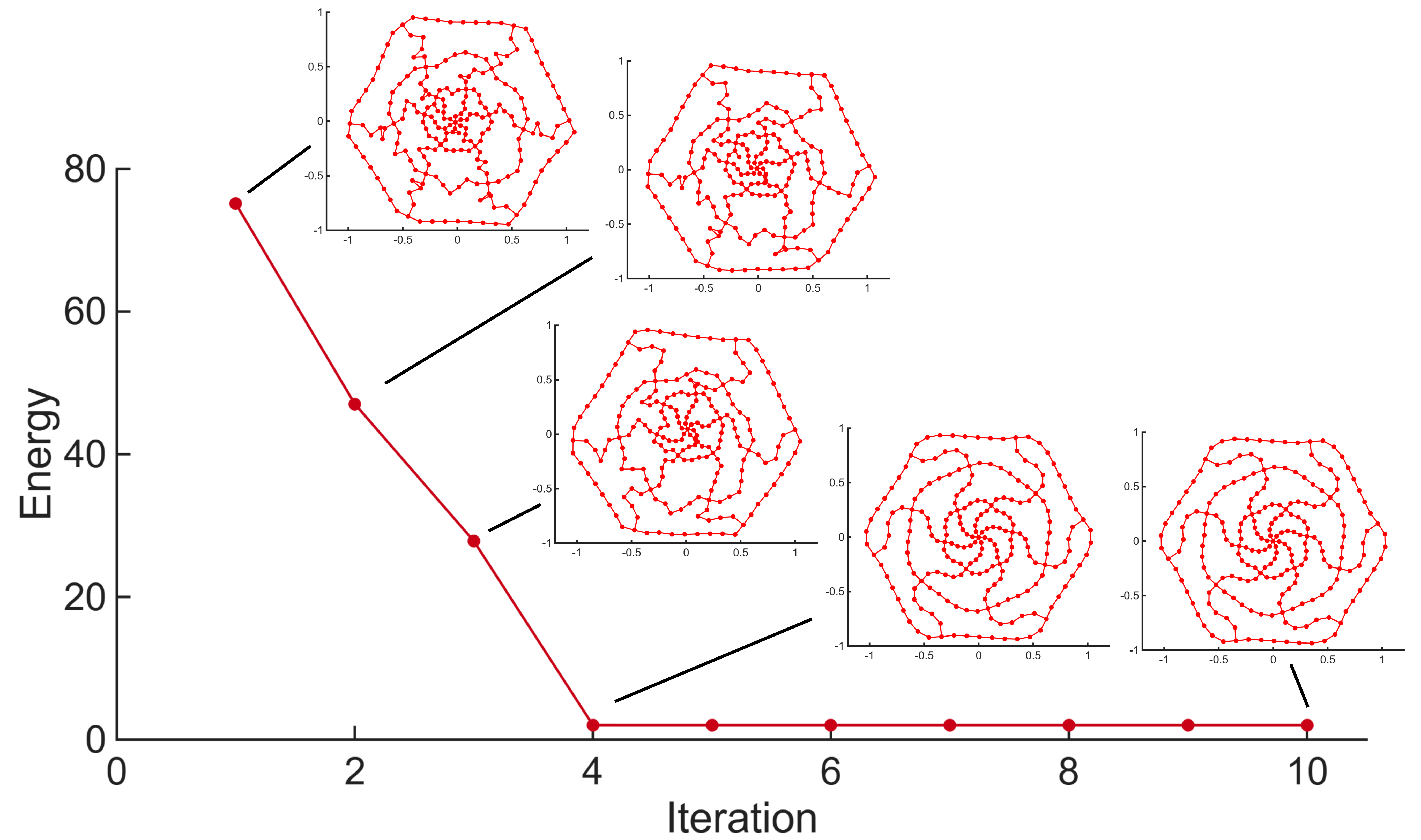}
    \caption{\textbf{The energy plot for the example shown in Fig.~\ref{fig:illustration}.} Here, we plot the energy $E$ in Eq.~\eqref{eqt:objective_function} versus the iteration number. The insets show the intermediate configurations at different iterations. It can be observed that the result converges rapidly.}
    \label{fig:energy}
\end{figure}

Besides the above example, Fig.~\ref{fig:results} (top row) shows three other examples of rod-based structures with different geometries and curvature properties. From the corresponding planar embedding results (bottom row), it can be observed that the rod segment lengths, angles at the major joints, and the bijectivity are all well-preserved. 

\begin{figure}[t]
    \centering
    \includegraphics[width=\linewidth]{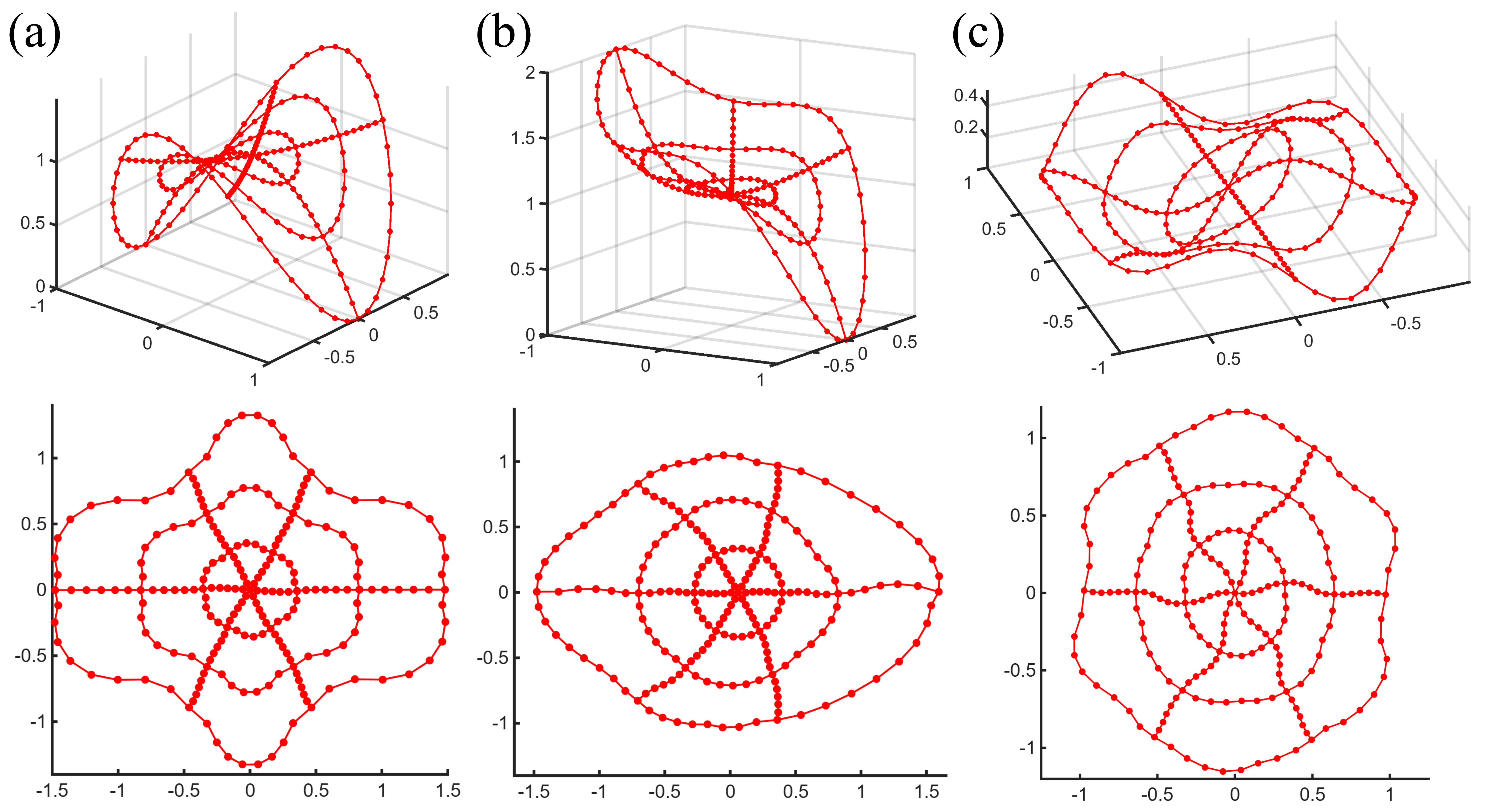}
    \caption{\textbf{Three examples of rod-based structures with different geometries (top) and the corresponding planar embeddings (bottom).} (a) A doubly curved rod-based structure. (b) A rod-based structure with a more prominent height variation. (c) A rod-based structure with multiple peaks.}
    \label{fig:results}
\end{figure}

After demonstrating the effectiveness of the proposed method in handling the above examples, we consider some more complex examples. Specifically, in Fig.~\ref{fig:result_surface}(a) we consider a rod-based structure representing a cloth surface. In Fig.~\ref{fig:result_surface}(b), we further consider a rod-based structure representing a human face. For both examples, our method is capable of producing a low-distortion planar embedding without overlaps. We remark that for these two examples, while the rod-based structures represent a certain surface, they still only consist of rod segments and joints without any specific part being identified as a ``face''. Consequently, the planar embedding only preserves the rod lengths and angles between major joints, while other geometric properties of the ``faces'' such as the face diagonal length will not be automatically preserved. The extension of our method for handling structures with a specific ``surface region'' is discussed later in Section~\ref{sect:extension}.

\begin{figure}[t!]
    \centering
    \includegraphics[width=\linewidth]{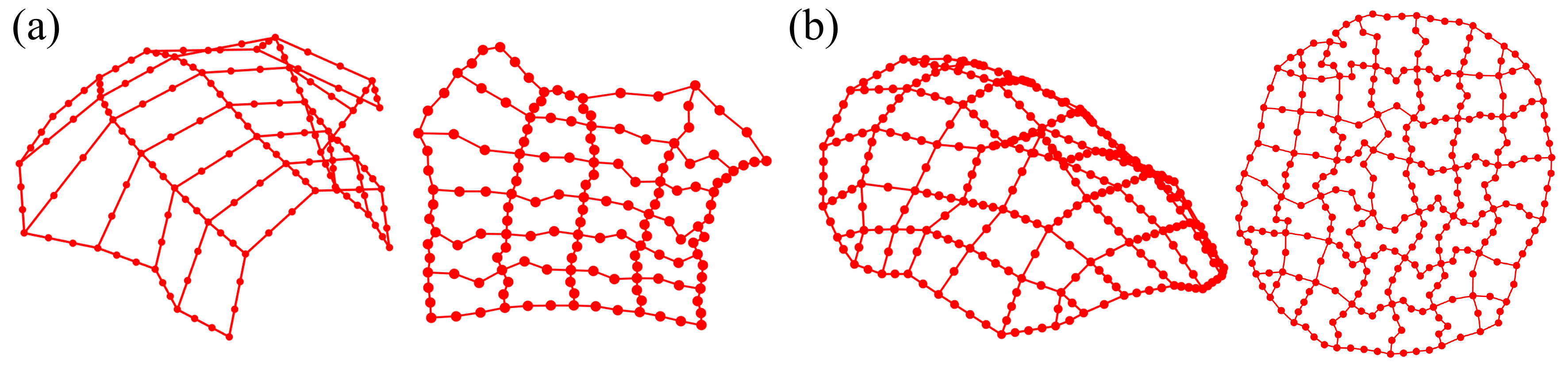}
    \caption{\textbf{Two examples of rod-based structures representing different surfaces and the corresponding planar embedding.} (a) The \emph{Cloth} model. (b) The \emph{Sophie} model.}
    \label{fig:result_surface}
\end{figure}

For a more quantitative analysis of the above experiments, Table~\ref{tab:results} shows the performance of our method (in terms of length error, angle error, and overlap number) on different examples. In all examples, we can see that the length and angle errors are very small, indicating that the embeddings exhibit low geometric distortion. We remark that in all examples, the average angle error is not as negligible as the length error, as the angle sum at a major joint at the interior of a 3D rod-based structure is affected by the curvature, while the angle sum at the corresponding point in the planar embedding is always $2\pi$. This discrepancy in curvature inevitably introduces certain angle errors. Besides, the overall length and angle errors are relatively large for the two surface examples (Fig.~\ref{fig:result_surface}(a) and Fig.~\ref{fig:result_surface}(b)) because of the large number of major joints between short rod segments. Specifically, the short rod segments unavoidably pose some limitations on the movement of the vertices on the plane throughout the optimization process, thereby causing a relatively large angle error when compared to the other examples. Besides, the number of overlaps is 0 for all examples, confirming the bijectivity of the embeddings.

\begin{table}[t]
    \centering
    \resizebox{1\linewidth}{!}{$
    \begin{tabular}{c|c|c|c}
        Example & Length error (Mean/SD) & Angle error (Mean/SD) & Overlap number \\ \hline
        Fig.~\ref{fig:illustration} 
        & $2.0 \times10^{-16}$ / $2.0 \times10^{-16}$
        & $1.4\times10^{-3}$ / $2.8\times10^{-3}$
        & 0\\
        Fig.~\ref{fig:results}(a) 
        & $2.9\times10^{-16}$ / $3.1\times10^{-16}$
        & $3.9\times10^{-4}$ / $1.1\times10^{-3}$
        & 0\\
        Fig.~\ref{fig:results}(b) 
        & $2.6\times10^{-16}$ / $2.7\times10^{-16}$
        & $7.1\times10^{-4}$ / $2.0\times10^{-3}$
        & 0\\
        Fig.~\ref{fig:results}(c) 
        & $2.8\times10^{-16}$ / $3.1\times10^{-16}$
        & $9.7\times10^{-4}$ / $3.7\times10^{-3}$
        & 0\\
        Fig.~\ref{fig:result_surface}(a) 
        & $5.6\times10^{-16}$ / $6.1\times10^{-16}$ 
        & $4.4\times10^{-2}$ / $1.6\times10^{-1}$
        & 0\\
        Fig.~\ref{fig:result_surface}(b)       
        & $5.3\times10^{-5}$ / $1.0\times10^{-5}$  
        & $3.1\times10^{-2}$ / $1.4\times10^{-1}$
        & 0
    \end{tabular}
    $}
    \caption{Performance of our method on different rod-based structures.}
    \label{tab:results}
\end{table}

\subsection{Efficiency and robustness of our method}
Besides, it is natural to study the efficiency and robustness of the method. In Fig.~\ref{fig:scaling}, we consider testing our method on several rod-based structures with the same overall geometry but different resolutions. It can be observed that the planar embedding results produced by our framework are highly similar, demonstrating the consistency of our method across different resolutions. In Table~\ref{tab:scaling_results}, we further quantify the performance for these examples. It can be observed that the length and angle errors are consistently low for all examples, and the number of overlaps is 0. Our method is also highly computationally efficient, generally taking only a few seconds for structures with hundreds of vertices. By performing a linear regression on $\log(|\mathcal{V}|)$ and $\log(\text{Time})$, we see that the computational time scales approximately as $O(|\mathcal{V}|^{1.9})$. More generally, for the IPOPT interior-point method solver~\cite{wachter2006implementation} used in solving our optimization problem, it is known that the cost depends on the sparsity pattern of the problem and hence may vary for different input rod-based structures.

\begin{figure}[t]
    \centering
    \includegraphics[width=\linewidth]{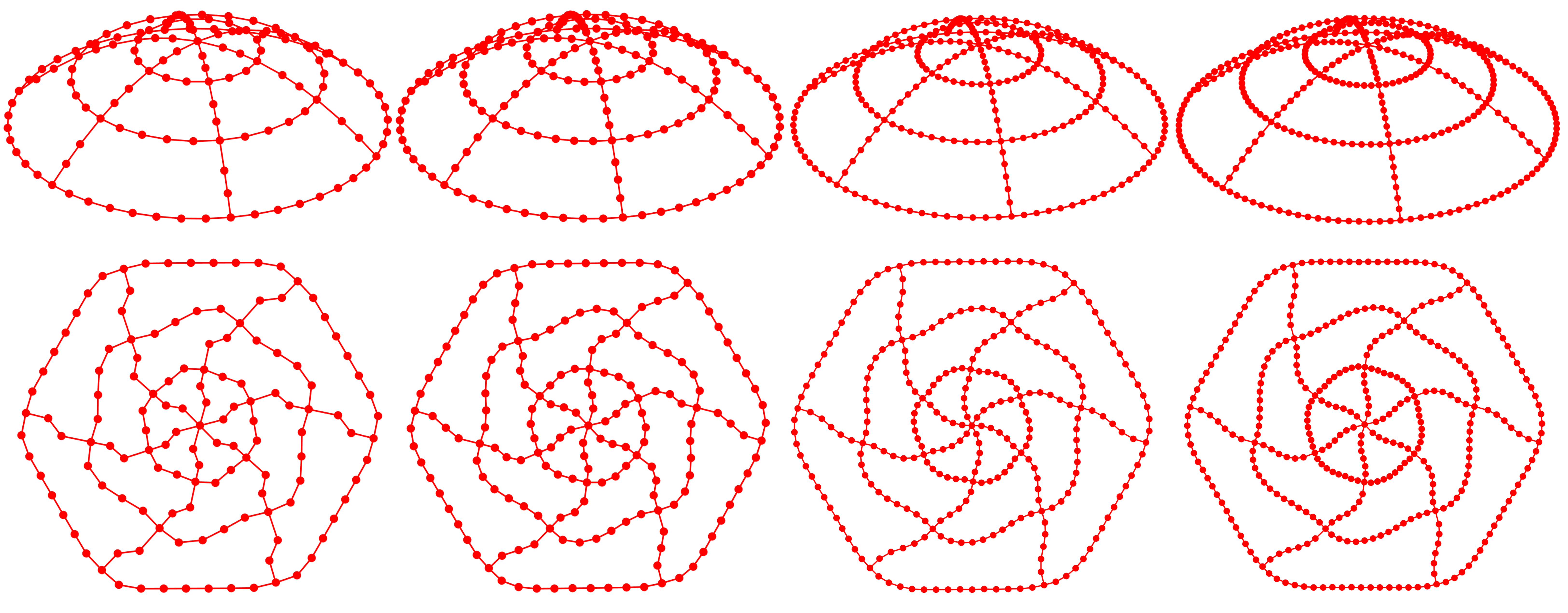}
    \caption{\textbf{Experimental results on rod-based structures with different resolutions.} The top row shows several rod-based structures with the same overall geometry but different resolutions. The bottom row shows the corresponding low-distortion planar embedding results. }
    \label{fig:scaling}
\end{figure}

\begin{table}[t]
    \centering
    \resizebox{1\linewidth}{!}{$
    \begin{tabular}{c|c|c|c|c}
        $|\mathcal{V}|$ & Time (s) &  Length error (Mean/SD) & Angle error (Mean/SD) & Overlap number\\ \hline

        133 & 1.8  
        & $1.5 \times10^{-16}$ / $1.7 \times10^{-16}$
        & $2.2\times10^{-3}$ / $1.1\times10^{-3}$
        & 0\\ 
        
        181 & 3.7  
        & $3.4 \times10^{-16}$ / $3.0 \times10^{-16}$
        & $1.2\times10^{-3}$ / $2.3\times10^{-3}$
        & 0\\

        289 & 7.6  
        & $2.8 \times10^{-16}$ / $2.7 \times10^{-16}$
        & $4.4\times10^{-4}$ / $8.5\times10^{-4}$
        & 0\\

        361 & 13.8
        & $3.7 \times10^{-16}$ / $3.5 \times10^{-16}$
        & $3.7\times10^{-4}$ / $6.7\times10^{-4}$
        & 0\\ 

    \end{tabular}
    $}
    \caption{Performance of our method on different rod-based structures with different resolutions. The four test cases here correspond to the four examples shown in Fig.~\ref{fig:scaling}.}
    \label{tab:scaling_results}
\end{table}

\begin{figure}[t!]
    \centering
    \includegraphics[width=0.95\linewidth]{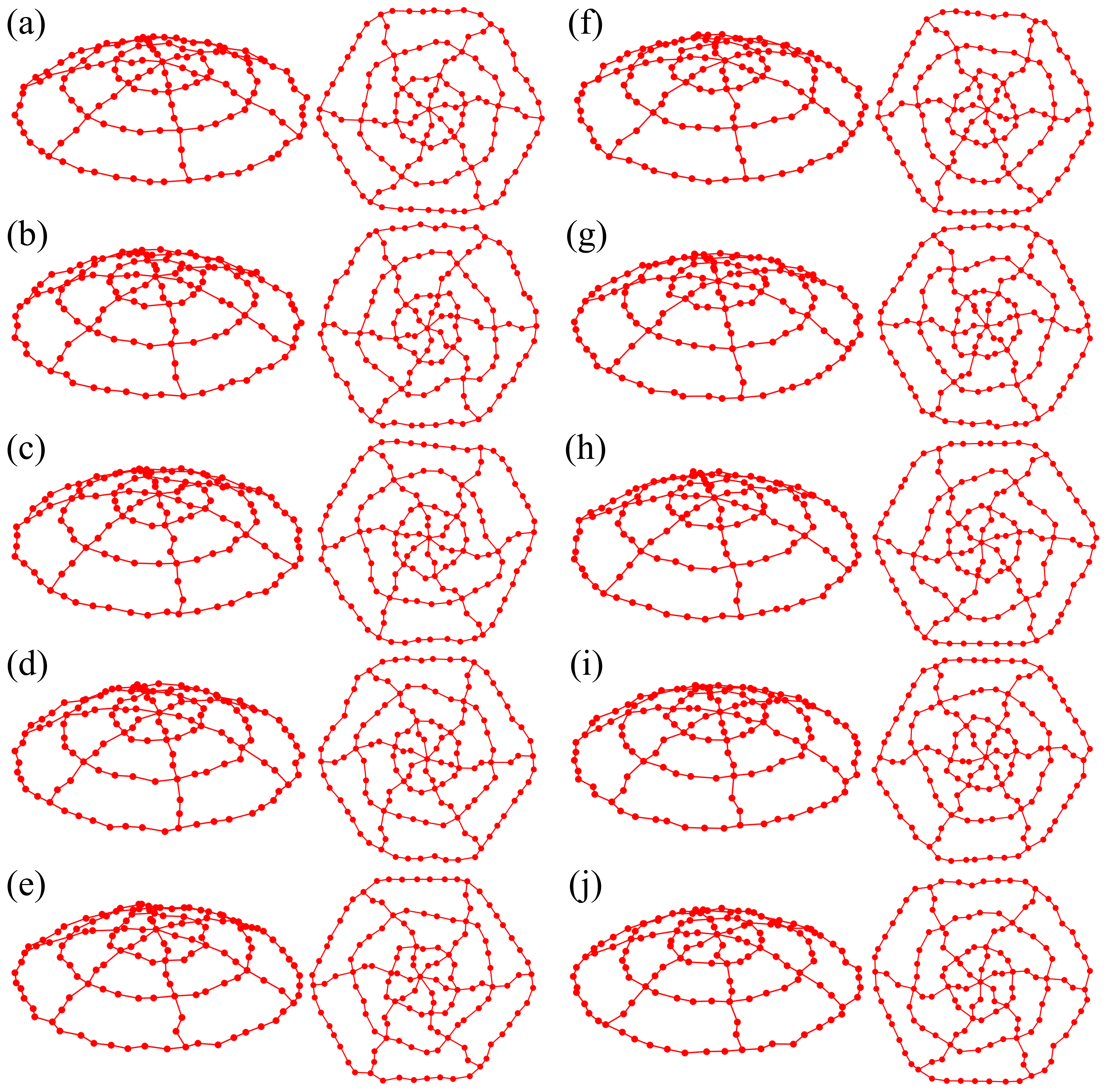}
    \caption{\textbf{Experimental results on rod-based structures with random noises.} (a)--(j) show 10 independent trials, where the left panel is the input 3D rod-based structure with 3\% uniformly distributed random noise added to the vertex coordinates and the right panel is the planar embedding result obtained by our method. }
    \label{fig:robustness}
\end{figure}

We then further consider testing the robustness of the method to small perturbations of the input rod-based structure. More specifically, we start with a given rod-based structure and consider adding 3\% uniformly distributed random noise to the coordinates of each vertex, i.e., $(x_i, y_i, z_i) \to (x_i + \Delta x_i, y_i + \Delta y_i, z_i + \Delta z_i)$, with the random numbers $\Delta x_i \in [0, 0.03 r_x]$, $\Delta y_i \in [0, 0.03 r_y]$, $\Delta z_i \in [0, 0.03 r_z]$ all sampled uniformly from the respective interval, where $r_x, r_y, r_z$ is the $x$-range, $y$-range, $z$-range of the rod-based structure. We then apply our proposed method to compute the low-distortion planar embedding and evaluate the performance. Fig.~\ref{fig:robustness} shows the results for 10 independent trials, from which we can see that the planar embedding results are highly consistent. As shown in Table~\ref{tab:robustness}, the length and angle errors of all embeddings are also consistently low, and all results are overlap-free. Altogether, the experimental results suggest that our method is highly robust.

\begin{table}[t]
    \centering
    \resizebox{1\linewidth}{!}{$
    \begin{tabular}{c|c|c|c}
        Example & Length error (Mean/SD) & Angle error (Mean/SD) & Overlap number \\ \hline
        Fig.~\ref{fig:robustness}(a)
        & $1.9 \times10^{-16}$ / $1.9 \times10^{-16}$
        & $2.2\times10^{-3}$ / $7.2\times10^{-3}$
        & 0\\
        Fig.~\ref{fig:robustness}(b)
        & $1.8 \times10^{-16}$ / $1.8 \times10^{-16}$
        & $2.3\times10^{-3}$ / $8.2\times10^{-3}$
        & 0\\
        Fig.~\ref{fig:robustness}(c)
        & $1.9 \times10^{-16}$ / $1.9 \times10^{-16}$
        & $2.8\times10^{-3}$ / $8.5\times10^{-3}$
        & 0\\
        Fig.~\ref{fig:robustness}(d)
        & $1.8 \times10^{-16}$ / $1.8 \times10^{-16}$
        & $4.4\times10^{-3}$ / $1.4\times10^{-2}$
        & 0\\
        Fig.~\ref{fig:robustness}(e)
        & $1.8 \times10^{-16}$ / $1.6 \times10^{-16}$
        & $5.5\times10^{-3}$ / $1.5\times10^{-2}$
        & 0\\
        Fig.~\ref{fig:robustness}(f) 
        & $1.7 \times10^{-16}$ / $1.8 \times10^{-16}$
        & $3.4\times10^{-3}$ / $1.2\times10^{-2}$
        & 0\\
        Fig.~\ref{fig:robustness}(g)   
        & $1.7 \times10^{-16}$ / $1.7 \times10^{-16}$
        & $3.8\times10^{-3}$ / $1.3\times10^{-2}$
        & 0\\
        Fig.~\ref{fig:robustness}(h) 
        & $1.8 \times10^{-16}$ / $1.8 \times10^{-16}$
        & $6.8\times10^{-3}$ / $2.4\times10^{-2}$
        & 0\\
        Fig.~\ref{fig:robustness}(i)   
        & $1.5 \times10^{-16}$ / $1.7 \times10^{-16}$
        & $1.7\times10^{-3}$ / $5.4\times10^{-3}$
        & 0\\
        Fig.~\ref{fig:robustness}(j)  
        & $1.6 \times10^{-16}$ / $1.6 \times10^{-16}$
        & $3.1\times10^{-3}$ / $1.1\times10^{-2}$
        & 0\\
    \end{tabular}
    $}
    \caption{Performance of our method on 10 rod-based structures with 3\% random noise added.}
    \label{tab:robustness}
\end{table}

\subsection{Comparison with the as-rigid-as-possible (ARAP) mapping}

\begin{figure}[t!]
    \centering
    \includegraphics[width=\linewidth]{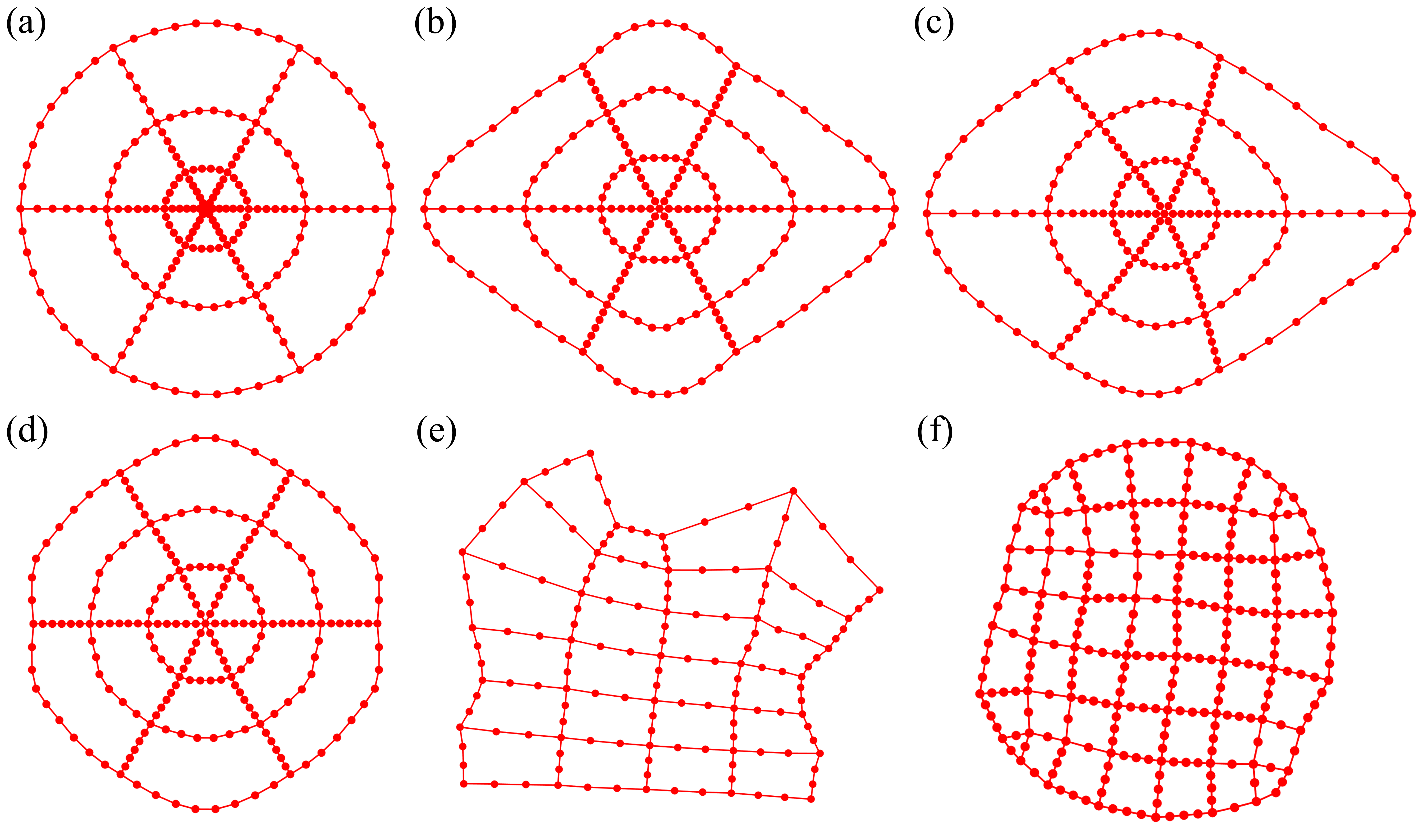}
    \caption{\textbf{The planar embedding results of several rod-based structures obtained by the as-rigid-as-possible (ARAP) mapping method~\cite{sorkine2007rigid,liu2008local}.} (a) The result of the structure in Fig.~\ref{fig:illustration}. (b)--(d) The results of the structures in Fig.~\ref{fig:results}(a)--(c) respectively. (e) The result of \emph{Cloth} model in Fig.~\ref{fig:result_surface}(a). (f) The result of the \emph{Sophie} model in Fig.~\ref{fig:result_surface}(b).}
    \label{fig:result_arap}
\end{figure}

Note that the task of representing a 3D object on a 2D domain is also widely studied in the field of surface parameterization, which often seeks a suitable mapping of a 3D discretized surface onto a planar domain with minimal geometric distortion. It is natural to ask how our proposed low-distortion planar embedding method compares with the existing surface parameterization methods when applied to the 3D rod-based structures. In particular, a parameterization method closely related to our proposed approach is the as-rigid-as-possible (ARAP) mapping method~\cite{sorkine2007rigid,liu2008local}, which aims to preserve the geometric details (including length and angle) of the mesh elements as much as possible. Therefore, here we focus on comparing our method with the ARAP method. 

For the computation of the ARAP mapping, we utilize the \texttt{arap} function available in the open-source geometry processing toolbox gptoolbox~\cite{gptoolbox}. Since the ARAP mapping method requires a triangulated surface mesh as input, we first triangulate each input 3D rod-based structure to obtain a simply-connected open triangulated surface. We then run the ARAP mapping method with the same initial 2D guess used in our framework. To allow the overall shape to deform freely on the plane as in our method, no boundary constraints are enforced in the ARAP mapping computation. After obtaining the ARAP mapping result for the triangulated surface, we extract its vertex coordinates together with the original connectivity information to form a planar embedding of the rod-based structure.

\begin{table}[t]
    \centering
    \resizebox{1\linewidth}{!}{$
    \begin{tabular}{c|c|c|c|c|c|c}
    \multirow{ 2}{*}{Example} & \multicolumn{3}{c|}{Our proposed method} & \multicolumn{3}{c}{As-rigid-as-possible (ARAP) mapping} \\ \cline{2-7}
     & Length error  & Angle error & Overlap \# & Length error & Angle error & Overlap \# \\ \hline
        Fig.~\ref{fig:illustration} 
        & $2.0 \times10^{-16}$  & $1.4\times10^{-3}$ & 0 & $1.5 \times 10^{-1} $& $9.2 \times 10^{-2}$ & 0\\
        Fig.~\ref{fig:results}(a) 
        & $2.9\times10^{-16}$ & $3.9\times10^{-4}$ & 0 & $5.4 \times 10^{-2}$ & $5.5 \times 10^{-2}$ & 0\\
        Fig.~\ref{fig:results}(b) 
        & $2.6\times10^{-16}$ & $7.1\times10^{-4}$ & 0 & $2.9\times10^{-2}$ & $3.6\times10^{-2}$ & 0\\
        Fig.~\ref{fig:results}(c) 
        &  $2.8\times10^{-16}$  & $9.7\times10^{-4}$  & 0 & $4.1\times10^{-2}$ & $8.0\times10^{-2}$ & 0\\
        Fig.~\ref{fig:result_surface}(a) 
        &  $5.6\times10^{-16}$  & $4.4\times10^{-2}$ & 0 & $5.1\times10^{-2}$  & $1.1\times10^{-1}$ & 0\\
        Fig.~\ref{fig:result_surface}(b)   
        &  $5.3\times10^{-5}$ & $3.1\times10^{-2}$ & 0 & $3.8\times10^{-2}$  & $5.1\times10^{-2}$ & 0
    \end{tabular}
    $}
    \caption{Comparison between our proposed method and the as-rigid-as-possible (ARAP) mapping method~\cite{sorkine2007rigid,liu2008local} on different rod-based structures. For each method, we compute the average length error for all rods, the average angle error for all angles at the major joints, and the number of overlaps.}
    \label{tab:results_comparison}
\end{table}

Fig.~\ref{fig:result_arap} shows the ARAP mapping results for several rod-based structures presented earlier. It can be observed that the results are generally consistent with those obtained by our proposed method and are overlap-free. One exception is the result in Fig.~\ref{fig:result_arap}(a), in which one can clearly see that the rods are highly squeezed in the central region in the mapping result. By further evaluating the length and angle errors (Table~\ref{tab:results_comparison}), one can see that our proposed method outperforms the ARAP method in preserving both rod lengths and angles by several orders of magnitude. A possible explanation is that the ARAP method treats the entire structure as a triangulated mesh, thereby also taking certain regions that are not important (such as the voids in the rod-based structures) in its consideration of the optimal low-distortion planar mapping. By contrast, our proposed method focuses exclusively on the vertices and rod segments of the input rod-based structure and can therefore effectively exploit the voids to handle the length- and angle-preserving constraints. Altogether, the comparison with the ARAP method has demonstrated the unique advantage of our proposed method for handling 3D rod-based structures.

\section{Extension to hybrid structures} \label{sect:extension}

In some cases, the structure may consist not only of one-dimensional rods but also of some surface regions. Besides preserving the geometry of the rods, one may also want to preserve the geometry of those surface regions in the embedding result as much as possible. Here, we consider extending our approach to handle the embedding of such hybrid structures. 

Specifically, for each surface region that is desired to be preserved, we can first include additional rod segments to form a local mesh representation. Then, we run our proposed algorithm with the geometry of these additional rods also considered. Because of the conditions required in our algorithm, the lengths and angles of these additional rods will also be preserved as much as possible. In other words, the optimization result will be a low-distortion planar embedding of the hybrid structure, with the isometric distortion (in both lengths and angles) reduced as much as possible.

\begin{figure}[t]
    \centering
    \includegraphics[width=\linewidth]{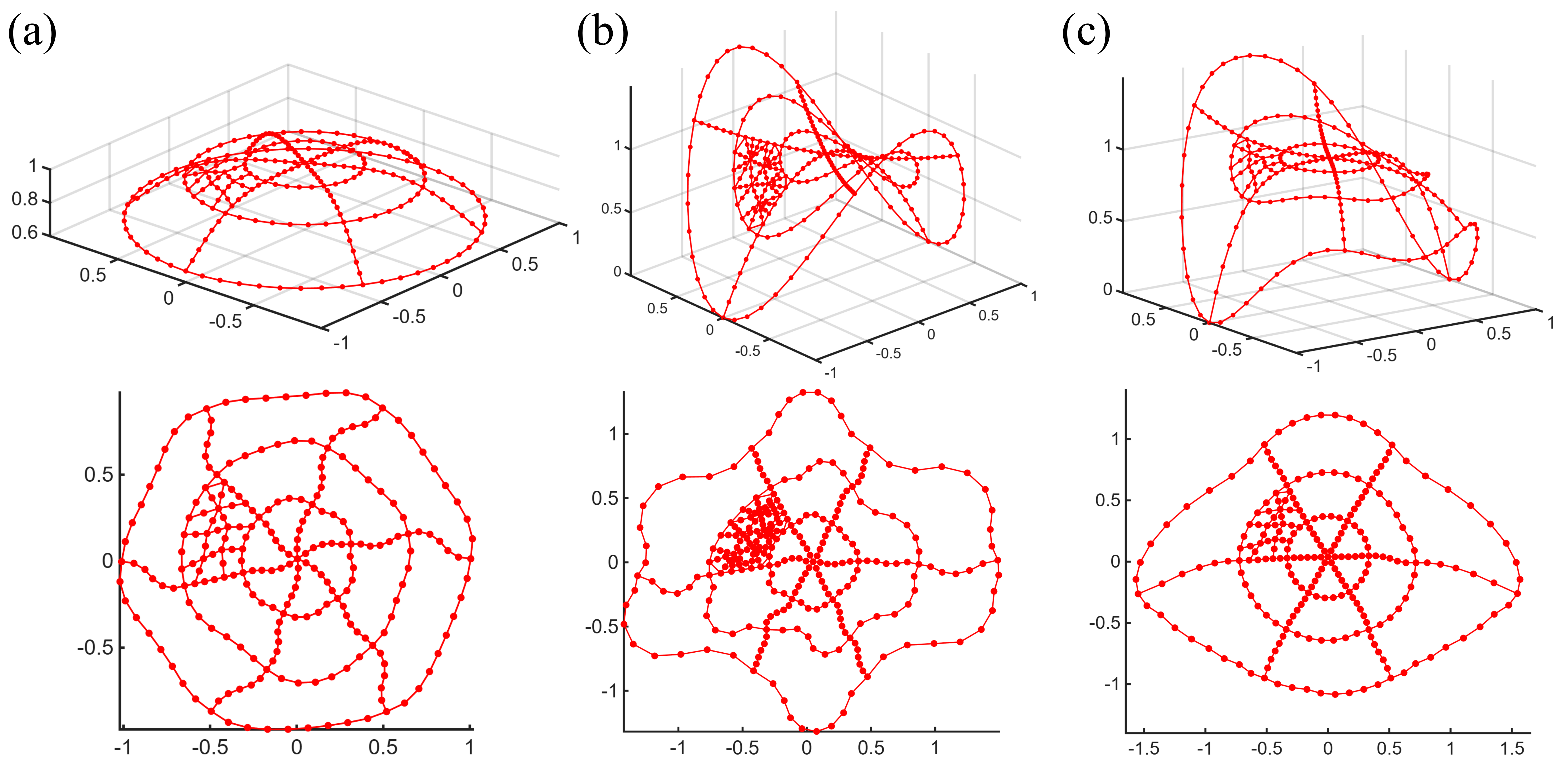}
    \caption{\textbf{Three examples of hybrid structures with different geometries (top) and the corresponding planar embeddings (bottom).} (a) A dome-shaped hybrid structure. (b) A doubly curved hybrid structure. (c) A hybrid structure with multiple peaks.}
    \label{fig:results_hybrid}
\end{figure}

\begin{table}[t]
    \centering
    \resizebox{1\linewidth}{!}{$
    \begin{tabular}{c|c|c|c}
        Example & Length error (Mean/SD) & Angle error (Mean/SD) & Overlap number \\ \hline
        Fig.~\ref{fig:results_hybrid}(a) 
        & $2.5 \times10^{-4}$ / $1.7 \times10^{-3}$
        & $5.6 \times10^{-3}$ / $1.5 \times10^{-2}$
        & 0\\ 
        Fig.~\ref{fig:results_hybrid}(b) 
        & $4.4 \times10^{-4}$ / $3.9 \times10^{-3}$
        & $3.7 \times10^{-3}$ / $1.2 \times10^{-2}$ 
        & 0\\ 
        Fig.~\ref{fig:results_hybrid}(c) 
        &  $7.2 \times10^{-4}$ / $4.3 \times10^{-3}$
        & $6.7 \times10^{-3}$ / $2.1 \times10^{-2}$
        & 0\\
    \end{tabular}
    $}
    \caption{Performance of our method on different hybrid structures.}
    \label{tab:results_hybrid}
\end{table}

To examine the performance of this approach, Fig.~\ref{fig:results_hybrid} shows three examples of hybrid structures with different geometries. First, in Fig.~\ref{fig:results_hybrid}(a) we consider a positively curved structure formed primarily by rods, with an additional surface region to be preserved. We then consider a doubly curved structure with an additional surface region as shown in Fig.~\ref{fig:results_hybrid}(b). Finally, we consider another curved structure with greater asymmetry and fluctuations in shape, again with an additional surface region to be preserved. For all these hybrid structure examples, our proposed algorithm is capable of producing the desired planar embeddings.

Table~\ref{tab:results_hybrid} shows the detailed performance analyses of our method on different hybrid structures. It can be observed that while the inclusion of the surface regions makes the optimization problem overconstrained, both the length and angle errors remain very small. Also, the number of overlaps is 0 in all examples. This demonstrates the effectiveness and generalizability of our proposed framework for handling a wider class of rod-based structures. We remark that one can further expand the coverage of the ``surface region'' and consider the case where the entire input is a triangle or quad mesh. For the case of triangle meshes, all triangle edges are considered as rod segments and all angles of all triangle elements are considered as the major joint angles to be preserved. For quad meshes, one may further add one diagonal rod segment for each quadrilateral face to prevent shearing and apply the planar embedding method. The proposed algorithm will aim to produce a planar parameterization of the input mesh with the isometric distortion reduced as much as possible. However, since isometric mappings of a 3D shape onto the 2D plane are generally impossible to achieve, the length and angle constraints may not be fully satisfied.

\section{2D-to-3D morphing process of the planar embeddings} \label{sect:morphing}

As discussed earlier in this work, the low-distortion planar embeddings of rod-based structures produced by our proposed algorithm can be utilized for various applications. In particular, the planar embeddings can be regarded as simplified 2D representations of the original 3D rod-based structures, which facilitate their manufacturing and storage. It is therefore natural to ask whether we can restore the 3D structures from the planar embedding results. 

Here, we simulate the 2D-to-3D morphing process 
using the deployment simulation approach in~\cite{choi2019programming}, which considers a spring energy model to simulate the deployment process. More specifically, in the spring energy model, all rods in the 2D planar embedding result $(\mathcal{P},\mathcal{E_{\text{2D}}})$ are treated as linear springs. The rest length of the spring representing rod $[p_{i}, p_{j}]$ is given by the length of the rod segment, i.e., $L_{ij} = \|p_{i} -p_{j}\|$. We can then consider the entire deployment process to be within the time interval $[0,1]$ and track the time-dependent position $\mathbf{p}_i(t) \in \mathbb{R}^3$ of all vertices $p_i$, with $\mathbf{p}_i(t=0) = p_i$. We then apply a pulling force on a selected group of vertices iteratively, starting from the 2D planar configuration of the structure. Specifically, for these selected group of vertices $\{p_s\}$, we set $\mathbf{p}_s(t=1) = v_s$, where $v_s$ is the vertex position of the original input 3D rod-based structure. The trajectory of these selected vertices throughout the deployment process is then given by
\begin{equation}
    \mathbf{p}_s(t) = p_s + t(v_s - p_s).
\end{equation}
The rest of the structure is then also deployed following a spring energy minimization process, with the vertex coordinates $\mathbf{p}_1(t), \dots, \mathbf{p}_{|\mathcal{V}|}(t)$ determined by
\begin{equation}
    \min_{\mathbf{p}_1(t), \dots, \mathbf{p}_{|\mathcal{V}|}(t)} \frac{1}{|\mathcal{E_{\text{2D}}}|} \sum_{\{(i,j): [p_{i}, p_{j}] \in \mathcal{E_{\text{2D}}}\}} \left(\frac{\|\mathbf{p}_i(t)-\mathbf{p}_j(t)\|-L_{ij}}{L_{ij}}\right)^2,
\end{equation}
constrained by the prescribed position of the selected pulling points at different time $t$. The deployment ends when the selected pulling points reach the final position of the 3D rod-based structure, i.e., $t=1$. In other words, iteratively solving the above spring energy minimization problem at different time $t\in [0,1]$ gives a continuous deployment path from the 2D configuration to a final 3D configuration, where the selected pulling points largely control the motion and the remaining points will follow them naturally, with the position optimized to reduce the distortion in the actual rod length. The springs provide flexibility in simulating deployment without precisely specifying how all points move throughout the process.

To demonstrate this idea, we consider the 2D-to-3D deployment process of the rod-based example in Fig.~\ref{fig:illustration}. In Fig.~\ref{fig:results_deployment}, we show snapshots of the deployment of the structure from the planar embedding state produced by our algorithm to the final 3D configuration. Here, we select only the rod intersection points (in black) as the pulling points and apply the above-mentioned spring energy model to simulate the deployment over time. It can be observed that the structure effectively morphs from the 2D state into the final 3D state, which matches the shape of the original 3D rod-based structure very well. From this experiment, we can see that the planar embeddings produced by our proposed algorithm can be effectively used in practical applications.

\begin{figure}[t]
    \centering
    \includegraphics[width=\linewidth]{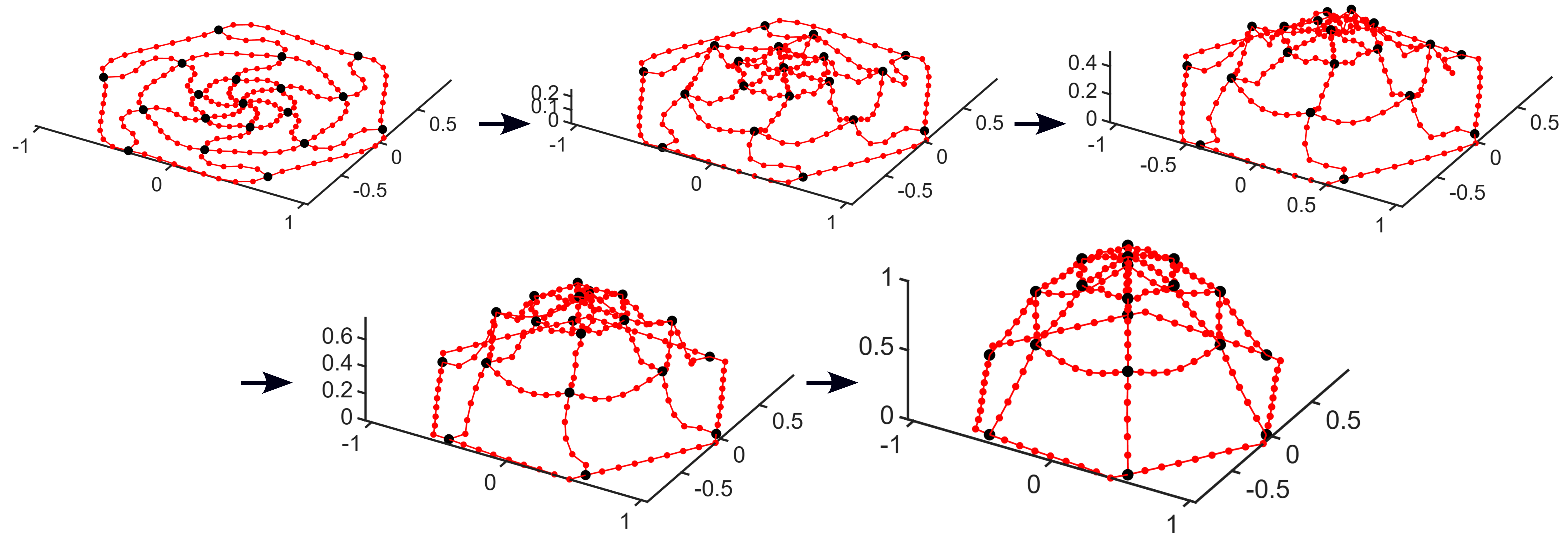}
    \caption{\textbf{Simulated deployment process of a rod-based structure from the planar configuration obtained by our low-distortion embedding method (top left) to the final 3D configuration (bottom right).} The black nodes indicate the pulling points. }
    \label{fig:results_deployment}
\end{figure}

\section{Conclusion} \label{sect:conclusion}

In this work, we have developed a novel method for the planar embedding of 3D structures composed of rod segments. Specifically, our method preserves the rod segment length and the intersection angles between the rods, thereby maintaining the key geometrical properties of the rod-based structures. Also, our method effectively prevents overlaps between rod segments in the planar embedding, thereby facilitating the practical fabrication and use of the planar representations. We have demonstrated the applicability of the method to a wide range of rod-based structures with different shapes. We have also demonstrated the feasibility of extending the formulation to hybrid structures. Using a simple mechanical model, we have further shown that the deployment from the 2D embedding to the desired 3D shape can be easily achieved. Altogether, our work paves a new way to the representation and simplification of rod-based structures.

In our future work, we plan to extend our approach and consider a broader range of rod-based structures with spatially varying elasticity conditions, allowing some parts of the structure to become more flexible and exhibit greater geometric distortions. To achieve this, one possible approach is to turn the equality constraints on lengths and angles in our current optimization problem into inequality constraints, imposing different bounds on different edges and angles to control their distortions. We also plan to investigate how the convergence behavior of the proposed algorithms is affected by relaxing the length- and angle-preservation conditions.

\section*{Declaration of Competing Interest}
The authors declare no conflict of interest.
  
\section*{Acknowledgements}
We thank Dr.~Mahmoud Shaqfa (ETH Zurich) for useful discussions. 

\bibliographystyle{elsarticle-num}
\bibliography{reference}

\begin{thebibliography}{10}
\expandafter\ifx\csname url\endcsname\relax
  \def\url#1{\texttt{#1}}\fi
\expandafter\ifx\csname urlprefix\endcsname\relax\def\urlprefix{URL }\fi
\expandafter\ifx\csname href\endcsname\relax
  \def\href#1#2{#2} \def\path#1{#1}\fi

\bibitem{wang2023rectifying}
B.~Wang, H.~Wang, E.~Schling, H.~Pottmann, Rectifying strip patterns, ACM Transactions on Graphics 42~(1) (2023) 1--18.

\bibitem{sydney2016biomimetic}
A.~Sydney~Gladman, E.~A. Matsumoto, R.~G. Nuzzo, L.~Mahadevan, J.~A. Lewis, Biomimetic {4D} printing, Nature Materials 15~(4) (2016) 413--418.

\bibitem{risso2022highly}
G.~Risso, M.~Sakovsky, P.~Ermanni, A highly multi-stable meta-structure via anisotropy for large and reversible shape transformation, Advanced Science 9~(26) (2022) 2202740.

\bibitem{liu2023deployable}
D.~Liu, D.~Pellis, Y.-C. Chiang, F.~Rist, J.~Wallner, H.~Pottmann, Deployable strip structures, ACM Transactions on Graphics 42~(4) (2023) 1--16.

\bibitem{wang2004level}
X.~Wang, Y.~Mei, M.~Y. Wang, Level-set method for design of multi-phase elastic and thermoelastic materials, International Journal of Mechanics and Materials in Design 1~(3) (2004) 213--239.

\bibitem{baek2018form}
C.~Baek, A.~O. Sageman-Furnas, M.~K. Jawed, P.~M. Reis, Form finding in elastic gridshells, Proceedings of the National Academy of Sciences 115~(1) (2018) 75--80.

\bibitem{qin2020genetic}
L.~Qin, W.~Huang, Y.~Du, L.~Zheng, M.~K. Jawed, Genetic algorithm-based inverse design of elastic gridshells, Structural and Multidisciplinary Optimization 62 (2020) 2691--2707.

\bibitem{guseinov2020programming}
R.~Guseinov, C.~McMahan, J.~P{\'e}rez, C.~Daraio, B.~Bickel, Programming temporal morphing of self-actuated shells, Nature Communications 11~(1) (2020) 237.

\bibitem{martin2021surface}
A.~Mart{\'\i}n-Pastor, F.~Gonz{\'a}lez-Quintial, Surface discretisation with rectifying strips on geodesics, Nexus Network Journal 23~(3) (2021) 565--582.

\bibitem{panetta2021computational}
J.~Panetta, F.~Isvoranu, T.~Chen, E.~Si{\'e}fert, B.~Roman, M.~Pauly, Computational inverse design of surface-based inflatables, ACM Transactions on Graphics 40~(4) (2021) 1--14.

\bibitem{ren2024computational}
Y.~Ren, J.~Panetta, S.~Suzuki, U.~Kusupati, F.~Isvoranu, M.~Pauly, Computational homogenization for inverse design of surface-based inflatables, ACM Transactions on Graphics 43~(4) (2024) 1--18.

\bibitem{schling2022designing}
E.~Schling, H.~Wang, S.~Hoyer, H.~Pottmann, Designing asymptotic geodesic hybrid gridshells, Computer-Aided Design 152 (2022) 103378.

\bibitem{floater2005surface}
M.~S. Floater, K.~Hormann, Surface parameterization: a tutorial and survey, Advances in multiresolution for geometric modelling (2005) 157--186.

\bibitem{sheffer2007mesh}
A.~Sheffer, E.~Praun, K.~Rose, Mesh parameterization methods and their applications, Foundations and trends in computer graphics and vision 2~(2) (2007) 105--171.

\bibitem{levy2002least}
B.~L{\'e}vy, S.~Petitjean, N.~Ray, J.~Maillot, Least squares conformal maps for automatic texture atlas generation, ACM Transactions on Graphics 21~(3) (2002) 362--371.

\bibitem{desbrun2002intrinsic}
M.~Desbrun, M.~Meyer, P.~Alliez, Intrinsic parameterizations of surface meshes, Computer Graphics Forum 21~(3) (2002) 209--218.

\bibitem{mullen2008spectral}
P.~Mullen, Y.~Tong, P.~Alliez, M.~Desbrun, Spectral conformal parameterization, Computer Graphics Forum 27~(5) (2008) 1487--1494.

\bibitem{jin2008discrete}
M.~Jin, J.~Kim, F.~Luo, X.~Gu, Discrete surface {R}icci flow, IEEE Transactions on Visualization and Computer Graphics 14~(5) (2008) 1030--1043.

\bibitem{choi2021efficient}
G.~P.~T. Choi, Efficient conformal parameterization of multiply-connected surfaces using quasi-conformal theory, Journal of Scientific Computing 87~(3) (2021) 70.

\bibitem{zou2011authalic}
G.~Zou, J.~Hu, X.~Gu, J.~Hua, Authalic parameterization of general surfaces using {L}ie advection, IEEE Transactions on Visualization and Computer Graphics 17~(12) (2011) 2005--2014.

\bibitem{zhao2013area}
X.~Zhao, Z.~Su, X.~D. Gu, A.~Kaufman, J.~Sun, J.~Gao, F.~Luo, Area-preservation mapping using optimal mass transport, IEEE Transactions on Visualization and Computer Graphics 19~(12) (2013) 2838--2847.

\bibitem{choi2018density}
G.~P.~T. Choi, C.~H. Rycroft, Density-equalizing maps for simply connected open surfaces, SIAM Journal on Imaging Sciences 11~(2) (2018) 1134--1178.

\bibitem{sorkine2007rigid}
O.~Sorkine, M.~Alexa, et~al., As-rigid-as-possible surface modeling, Symposium on Geometry Processing 4 (2007) 109--116.

\bibitem{liu2008local}
L.~Liu, L.~Zhang, Y.~Xu, C.~Gotsman, S.~J. Gortler, A local/global approach to mesh parameterization, Computer Graphics Forum 27~(5) (2008) 1495--1504.

\bibitem{wang2018novel}
Z.~Wang, Z.~Luo, J.~Zhang, E.~Saucan, A novel local/global approach to spherical parameterization, Journal of Computational and Applied Mathematics 329 (2018) 294--306.

\bibitem{choi2022recent}
G.~P.~T. Choi, L.~M. Lui, Recent developments of surface parameterization methods using quasi-conformal geometry, Handbook of Mathematical Models and Algorithms in Computer Vision and Imaging: Mathematical Imaging and Vision (2022) 1--41.

\bibitem{pan2022constructing}
M.~Pan, F.~Chen, Constructing planar domain parameterization with {HB}-splines via quasi-conformal mapping, Computer Aided Geometric Design 97 (2022) 102133.

\bibitem{pan2023g1}
M.~Pan, R.~Zou, W.~Tong, Y.~Guo, F.~Chen, $g^1$-smooth planar parameterization of complex domains for isogeometric analysis, Computer Methods in Applied Mechanics and Engineering 417 (2023) 116330.

\bibitem{kapl2018construction}
M.~Kapl, G.~Sangalli, T.~Takacs, Construction of analysis-suitable $g^1$ planar multi-patch parameterizations, Computer-Aided Design 97 (2018) 41--55.

\bibitem{zou2025mat}
R.~Zou, M.~Pan, Y.~Zheng, F.~Chen, W.~Tong, {MAT}-parameterization: Volumetric multi-patch parameterizations of complex domains for isogeometric analysis using {MAT}-based decomposition, Computer Methods in Applied Mechanics and Engineering 445 (2025) 118187.

\bibitem{tutte1963draw}
W.~T. Tutte, How to draw a graph, Proceedings of the London Mathematical Society 3~(1) (1963) 743--767.

\bibitem{choi2019programming}
G.~P.~T. Choi, L.~H. Dudte, L.~Mahadevan, Programming shape using kirigami tessellations, Nature Materials 18~(9) (2019) 999--1004.

\bibitem{wachter2006implementation}
A.~W{\"a}chter, L.~T. Biegler, On the implementation of an interior-point filter line-search algorithm for large-scale nonlinear programming, Mathematical Programming 106~(1) (2006) 25--57.

\bibitem{interx}
{NS}, Curve intersections, {MATLAB} {C}entral {F}ile {E}xchange, \url{https://www.mathworks.com/matlabcentral/fileexchange/22441-curve-intersections}, accessed on June 26, 2025 (2010).

\bibitem{gptoolbox}
A.~Jacobson, et~al., {gptoolbox}: Geometry processing toolbox, http://github.com/alecjacobson/gptoolbox (2024).

\end{thebibliography}

\appendix
\section{Proof of Theorem 1}
    By assumption, we have $\sum_{i = 1}^n  \text{Area}(T_i) = \text{Area}(\mathcal{B})$ for the shape optimization result.
    
    Now, suppose there is still at least one overlap in the shape optimization result. Then, there exist two triangles $T_k$ and $T_j$ such that $T_k \bigcap T_j \neq \emptyset$. Let $\mathcal{O} = T_k\bigcap T_j$. Since $\mathcal{O}$ is non-empty, we have $\text{Area}(\mathcal{O})>0$. 
    
    In the summation $\sum\limits_{i = 1}^n  \text{Area}(T_i)$, $\text{Area}(\mathcal{O})$ is counted in both $\text{Area}(T_k)$ and $\text{Area}(T_j)$. Therefore, 
    \begin{equation}
       \text{Area}(\mathcal{B})= \sum_{i = 1}^n  \text{Area}(T_i) = \text{Area}(\mathcal{O}) + \text{Area}(\mathcal{B})>\text{Area}(\mathcal{B}),
    \end{equation}
    which leads to a contradiction. Therefore, the shape optimization result will not have any overlaps if $\sum_{i = 1}^n  \text{Area}(T_i) = \text{Area}(\mathcal{B})$.

\end{document}